\newif\iffull\fulltrue
\newif\ifshort\shortfalse
\newif\ifcomments\commentstrue
\newif\ifdavide\davidetrue

\documentclass[preprint]{elsarticle}

\usepackage{mathtools,amsmath,amsfonts,amssymb,amsthm,bbm,wasysym}
\usepackage{amscd,xy}
\usepackage{enumerate}
\usepackage{color,graphicx,xcolor}
\usepackage{hyperref}
\usepackage[top=2cm, bottom=1.3cm, left=3.5cm, right=3.5cm, marginparwidth=2.8cm]{geometry}
\usepackage{enumerate}
\usepackage{nameref}
\usepackage{stmaryrd}
\usepackage{mathpartir}
\usepackage{graphicx}
\usepackage{etex}
\usepackage{thmtools, thm-restate, amsthm, cleveref}

\theoremstyle{plain}
\newtheorem{theorem}{Theorem}
\newtheorem{proposition}[theorem]{Proposition}
\newtheorem{lemma}[theorem]{Lemma}
\newtheorem{corollary}[theorem]{Corollary}

\theoremstyle{definition}
\newtheorem{definition}[theorem]{Definition}
\newtheorem{example}[theorem]{Example}
\newtheorem{remark}[theorem]{Remark}

\usepackage{DSarrow} 
\input{commonmacros.sty}
\input{adrien.sty}
\input{cbv.sty}
\input{encodings.sty}
\input{equations.sty}
\input{davide.sty}
\input{preorders.sty}

\newcommand{\rref}[1]{(\ref{#1})}

\title{Eager Functions as Processes}
\author[1]{Adrien Durier}
\ead{adrien.durier@ens-lyon.fr}

\author[2]{Daniel Hirschkoff}
\ead{daniel.hirschkoff@ens-lyon.fr}

\author[3]{Davide Sangiorgi}
\ead{davide.sangiorgi@cs.unibo.it}

\affiliation[1]{organization={École Normale Supérieure de Lyon},
  addressline={46, all{\'e}e d'Italie}, postcode={69007}, city={Lyon},
  country={France}}
\affiliation[2]{organization={École Normale Supérieure de Lyon},
  addressline={46, all{\'e}e d'Italie}, postcode={69007}, city={Lyon},
  country={France}}
\affiliation[3]{organization={Università di Bologna and INRIA},
  addressline={Mura Anteo Zamboni, 7}, postcode={40126}, city={Bologna},
  country={Italy}}

\fntext[fn3]{Sangiorgi  acknowledges support from the
  MIUR-PRIN project `Analysis of
Program Analyses' (ASPRA, ID: \texttt{201784YSZ5\_004}), and from the
European Research
Council (ERC) Grant DLV-818616 DIAPASoN.}

\date{}

\begin{document}

\begin{abstract}
  We study Milner's encoding of the call-by-value $\lambda$-calculus
  into the $\pi$-calculus.  We show that, by tuning the encoding to two
  subcalculi of the $\pi$-calculus (Internal $\pi$ and Asynchronous Local
  $\pi$), the equivalence on $\lambda$-terms induced by the encoding
  coincides with Lassen's eager normal-form bisimilarity, extended to
  handle $\eta$-equality.  As behavioural equivalence in the
  $\pi$-calculus we consider contextual equivalence and barbed
  congruence. We also extend the results to preorders.

  A crucial technical ingredient in the proofs is the
  recently-intro\-du\-ced technique of unique solutions of equations,
  further developed in this paper.  In this respect, the paper also
  intends to be an extended case study on the applicability and
  expressiveness of the technique.




\end{abstract}


\maketitle

\paragraph{Acknowledgements}
We are delighted to be able to contribute to the Festschrift in  honour 
of Fu Yuxi.
We would like to take this opportunity for heartily thanking him  for
  his many technical and scientific  contributions,  for
his work in favour of the  concurrency theory community, both locally in 
Shanghai and elsewhere, and  for all the discussions and time  spent 
together.

\section*{Introduction}
  Milner's work on functions as processes~\cite{milner:inria-00075405,encodingsmilner},
that shows  how the evaluation strategies of  {\em call-by-name 
  $\lambda$-calculus} and   {\em call-by-value
  $\lambda$-calculus}~\cite{Abr88,DBLP:journals/tcs/Plotkin75} can be faithfully mimicked in
the $\pi$-calculus,  
 is generally considered a
landmark in Concurrency Theory, and more generally in  Programming Language Theory.
 The comparison with the $\lambda$-calculus 
   is a significant  expressiveness test for
the $\pi$-calculus. 
More than that, 
it
promotes the $\pi$-calculus to be  a basis for general-purpose 
 programming languages in which communication is the fundamental
 computing primitive. 
From the $\lambda$-calculus point of view, the comparison  provides the means  to
study
 $\lambda$-terms
 in    contexts  other  than  purely  sequential ones,
 and with the  instruments available to reason about 
 processes.
Further, Milner's work, and the works that followed it, have contributed to
understanding and developing the theory of the $\pi$-calculus.

More precisely,   Milner shows the operational correspondence between
reductions in the $\lambda$-terms and in the encoding $\pi$-terms. 
He then uses the correspondence to prove that the encodings are
\emph{sound}, i.e., if the processes encoding 
 two $\lambda$-terms  are
behaviourally equivalent, then the source $\lambda$-terms are also
behaviourally equivalent in the $\lambda$-calculus.
Milner also shows that the converse,   \emph{completeness}, fails, intuitively because 
the  encodings allow one to test the $\lambda$-terms in all  contexts
of the $\pi$-calculus~--- more
diverse than those of the $\lambda$-calculus. 

The main problem that Milner's work left open is the characterisation
of the equivalence  on $\lambda$-terms induced by the encoding, whereby
two $\lambda$-terms are equal if their  encodings are behaviourally
equivalent $\pi$-calculus terms. 
The question is largely independent of the precise form of 
 behavioural equivalence adopted in the $\pi$-calculus
because the encodings are deterministic (or at
least confluent). In the 
paper we consider contextual equivalence (that coincides with may
testing and trace equivalence) and barbed congruence (that coincides
with bisimilarity). 


For the call-by-name  $\lambda$-calculus, the answer was found shortly later
\cite{San93,cbn}: 
the equality induced is the equality of Levy-Longo 
 Trees~\cite{LONGO1983153}, the lazy
variant of  B{\"o}hm Trees.  
It is actually also possible to 
obtain B{\"o}hm Trees, by modifying the call-by-name encoding so to allow also reductions
underneath a $\lambda$-abstraction, and by including divergence among the observables \cite{xian}. 
These results show that, at least for call-by-name, the $\pi$-calculus encoding, while not
fully abstract for the contextual equivalence of the $\lambda$-calculus, is in remarkable
agreement with the theory of the $\lambda$-calculus: several well-known models of the
$\lambda$-calculus yield   Levy-Longo Trees
 or B{\"o}hm  Trees  as
their induced equivalence~\cite{levy75,LONGO1983153,barendregt1984lambda}.

For call-by-value, in contrast, the problem of identifying the equivalence induced by the
encoding  has remained open, for two
main reasons. First, tree structures in call-by-value are less studied
and less established  than in call-by-name. Secondly, proving
completeness of an encoding of $\lambda$ into $\pi$ requires sophisticated proof techniques. For
call-by-name, for instance, a central role is played by 
\emph{bisimulation up-to contexts}. For call-by-value, however,
existing proof techniques, including `up-to contexts',
 appeared not to be  powerful enough.  


In this paper we study the above open problem for call-by-value. 
Our main result is that the equivalence induced on $\lambda$-terms by their call-by-value 
encoding into the $\pi$-calculus  is 
\emph{eager normal-form bisimilarity} \cite{lassentrees,lassentrees2}. 
This is a tree structure for call-by-value,
proposed  by Lassen as the call-by-value  counterpart 
   of
Levy-Longo Trees. 
Precisely we obtain the variant that is insensitive to some 
$\eta$-expansions, called \emph{$\eta$-eager
normal-form bisimilarity}. It validates the \etaexp law for variables:

\begin{equation}
\abs y xy = x
\end{equation}

This law is also valid for abstractions: $\abs y (\abs z M)y = \abs z
M$ if $y$ does not occur free in $M$.
However, in a weak call-by-value setting, $\eta$-expanded terms should not always be equated: 
indeed, $\Omega$ diverges, while $\abs x \Omega x$ converges to a value.

To obtain the results we have however to make a few adjustments
to Milner's encoding and/or specialise the target language of the encoding.
 These adjustments have to do with the presence 
 of free
outputs 
(outputs of known names) in the encoding. Indeed, 
Milner had initially
translated call-by-value  $\lambda$-variables using a free output: 
the translation of the variable $x$ would be a free output 
$\out p x$, where $p$ is the continuation, or location, at which $x$ is 
evaluated.
In the original encoding, therefore, the encoding of $x$ at $p$ (the
encoding of $\lambda$-terms is parametrised upon a name) is defined as
$\out px$, which can be written as follows:
\begin{equation}
\label{e:VarOne}
\encma x  p \defi \out px
.
\end{equation}

 However this 
 is troublesome for 
 the validity of $\betav$-reduction 
(the property that $\lambda$-terms that are related by 
$\betav$-reduction~--- the call-by-value $\beta$-reduction~--- 
 are also equal in the $\pi$-calculus).
 Milner solved the problem by ruling out the initial
free output $\out p x$ and by replacing it with a bound output $\new y\out p y$ followed 
by a static link $ \alpilink y x$. A static link $\alpilink y x$ forwards any 
name received by $y$ to $x$, therefore acting as a substitution between 
$x$ and $y$ (while also constraining the behaviour of the context, that may only 
use $y$ in output).
Thus, in the modified encoding, we have:

\begin{equation}
\label{e:VarMore}
\encmap x p \defi \new y (\out p y. \alpilink y x)
.
\end{equation}
It was indeed shown later \cite{sangiorgiphd} that  with \reff{e:VarOne}
the validity of $\betav$-reduction fails. 
Accordingly, the final journal paper~\cite{encodingsmilner} does not even mention
encoding~\reff{e:VarOne}. 
If one wants to maintain the simpler rule  \reff{e:VarOne},  
then the validity of $\betav$-reduction can be regained  
by taking, as target language, a subset of the  $\pi$-calculus
in which only the output capability of names is communicated. 
This can be enforced either by imposing a behavioural type system including capabilities 
\cite{iotypes}, or by  syntactically taking a dialect of the $\pi$-calculus in which only the
output capability of names is communicated, such as Local $\pi$~\cite{localpi}.

The encoding  \reff{e:VarMore} still  makes use of free
outputs~--- the final action of $\alpilink y x $ is a free output on $x$.
While  this limited  form of free
output is harmless for the validity of $\betav$-reduction, 
we show in the paper that 
this 
 brings problems when analysing $\lambda$-terms with free
variables: desirable call-by-value equalities  fail.
An example is given by the law:
\begin{equation}
  \label{eq:nonlaw}
  I(x\val)  = x\val
\end{equation}
where $I$ is $\abs zz$ and $\val$ is a value.

Law~\rref{eq:nonlaw} is valid in any model of \emph{call-by-value}, as 
any context making both terms closed also equates them: we have 
$\abs x I(x\val) = \abs x x\val$, for instance. 
It also holds in any theory of open call-by-value, as far as we know. 
This is indeed very natural: for any substitution of $x$ with a closed value, 
the terms become related by $\betav$
--- and the identity should never be made observable in a theory of the 
$\lambda$-calculus.


Two possible
 solutions to recover law~\rref{eq:nonlaw} are:

\begin{enumerate}
\item rule out  the free outputs; this essentially means transplanting the encoding 
onto the Internal $\pi$-calculus, \Intp~\cite{internalpi},  a version of the
  $\pi$-calculus in which any name emitted in an  output is fresh;

\item control the use of capabilities in the $\pi$-calculus; for
  instance taking Asynchronous Local
  $\pi$, \alpi~\cite{localpi}  as the  target of the translation. 
 (Controlling capabilities allows one to impose a directionality on
names, which, under certain
technical conditions, may  hide the identity of the emitted names.)
\end{enumerate} 

In the paper we consider both approaches, and show that 
in both cases, the equivalence induced coincides with
$\eta$-eager
normal-form bisimilarity.

In summary, there are two main contributions in the paper:
\begin{enumerate}
\item Showing that Milner's encoding fails to equate terms that should be equal in call-by-value.
\item Rectifying the encoding, by considering different target calculi, and  
investigating Milner's problem in such a setting.
\end{enumerate}
The rectification we make does not really change the essence of the encoding -- 
in one case, the encoding actually remains the same. Moreover, the languages used 
are well-known dialects of the \pc, studied in the literature for other reasons. 
In the encoding, they allow us to avoid certain accidental misuses of the names 
emitted in the communications. The calculi were not known at the time of Milner's 
paper \cite{encodingsmilner}.

A key role in the completeness proof is played by the technique of
\emph{unique solution of equations},  proposed in~\cite{usol}. 
The structure induced by Milner's 
call-by-value encoding was expected to look like Lassen's trees; 
however existing proof 
techniques did not seem powerful enough to prove it. 
The unique solution technique allows one to derive process bisimilarities from
equations whose infinite unfolding does not introduce divergences, by
proving that the processes are solutions of the same equations.  The
technique can be generalised to possibly-infinite systems of
equations, and can be  strengthened by allowing certain
  kinds of divergences in equations.  In this respect, another goal
of the paper is to carry out an extended case study on the
applicability and expressiveness of the techniques.  Then, a
by-product of the study are a few further developments of the
technique. 
In particular, one such result allows us to transplant uniqueness of
solutions from a system of equations, for which divergences
are easy to analyse, to another one. Another result is about
the application of the technique to preorders.

Finally, we consider 
preorders~--- thus referring to the
preorder on $\lambda$-terms induced by a behavioural preorder on their $\pi$-calculus
encodings. 
We introduce a preorder on Lassen's trees (preorders had not been considered by Lassen)
and show that this is the preorder on $\lambda$-terms induced by the call-by-value
encoding, when the behavioural relation on $\pi$-calculus terms is 
the ordinary contextual preorder (again, with the same restrictions as
mentioned above). 
With the move from equivalences to preorders, the overall structure of
the proofs of our full abstraction results remains the same. 
However, the impact on the application of the unique-solution
technique is substantial, because the phrasing of this technique in
the cases of preorders and of  equivalences is quite different.




\paragraph{Further related work}
The standard behavioural equivalence in the $\lambda$-calculus
is contextual equivalence. Encodings into the $\pi$-calculus
(be it for call-by-name or call-by-value) break  contextual equivalence
 because $\pi$-calculus contexts
are richer than those in the (pure) $\lambda$-calculus. In the paper
we try to understand how far beyond contextual equivalence the
discriminating power of the $\pi$-calculus brings us, for
call-by-value.   
The opposite approach is to restrict the set of 'legal' $\pi$-contexts
so to remain faithful to contextual equivalence. This approach has been
followed, for call-by-name, and using type systems, in
\cite{BHYseqpi,toninho:yoshida:esop18}. 

Open call-by-value  has been studied in~\cite{accattolicbv}, where the
focus is on operational properties of $\lambda$-terms; behavioural
equivalences are not considered.
An important difference with our work is that
in~\cite{accattolicbv},
$x\val_1\cdots\val_k$ is treated as a
value, i.e., $\beta$-reduction can be triggered when the argument has
this shape.



An extensive presentation of call-by-value, including denotational
models, 
is  Ronchi della Rocca and
Paolini's book~\cite{DBLP:series/txtcs/RoccaP04}.

In~\cite{usol}, the unique-solution technique is used in the 
completeness proof for Milner's call-by-name encoding. That proof 
essentially revisits the proof of~\cite{cbn}, which is based on 
bisimulation up-to context. We have  explained
 above that the case for
call-by-value  is quite different.


\paragraph{Structure of the paper} We recall basic definitions about
the call-by-value $\lambda$-calculus and the $\pi$-calculus in
Section~\ref{s:background}. The technique of unique solution of
equations is introduced in Section~\ref{s:usol}, together with some
new developments. Section~\ref{s:enc:cbv}
presents our analysis of Milner's encoding, beginning with the shortcomings
related to the presence of free outputs. 
{
 The first solution to these shortcomings is to move to the Internal
$\pi$-calculus: this is described in Section~\ref{s:enc:pii}. }
For the proof of completeness, in Section~\ref{s:complete},
we  rely on unique solution of equations; we also compare 
such technique with the `up-to techniques'.  
{
The second solution is to move to the Asynchronous Local
  $\pi$-calculus: this is discussed in 
Section~\ref{s:localpi}.}
We show in Section~\ref{s:contextual} how our
results can be adapted to   preorders and to contextual equivalence.
Finally in Section~\ref{s:concl}  we present  conclusions and
directions for future work.

\paragraph{Comparison with the results published
in~\cite{DBLP:conf/lics/DurierHS18}} This paper is an extended
version of~\cite{DBLP:conf/lics/DurierHS18}. We provide here
detailed proofs which were either absent or only sketched
in~\cite{DBLP:conf/lics/DurierHS18}, notably for: the soundness of 
the encoding into \Intp{} (Sections~\ref{s:oper:corr} and~\ref{s:sound}, 
\ref{a:soundness:pii}), completeness of the same encoding (Section~\ref{s:complete}), 
the unique solution technique for 
contextual relations and preorders 
(Section~\ref{ss:usol_pre}, \ref{a:usoltrace}), as 
well as some more details about the full abstraction proofs 
for contextual preorders (Section~\ref{ss:fa_preorders}) and the encoding into \alpi~(Section~\ref{s:fa:alpi}).
We also include more detailed discussions along the paper, notably
about Milner's encoding (Section~\ref{s:enc:cbv}), 
the encoding into \Intp~(Section~\ref{ss:enc_pii}), and 
the encoding into \alpi~(Section~\ref{alpiopsem}).




\section{Background material}
\label{s:background}
Throughout the paper,  $\R$ ranges over relations.
The composition of  two   relations
$\R$ and $\R'$  is written 
$\R \: \R'$.
We often use infix notation for relations; thus 
 $P \RR Q$ means ${(P, Q)}\in\R$. 
A tilde  represents  a tuple.  
The $i$-th element of a tuple $\til P$ is  referred to as $P_i$. 
 Our notations are extended to tuples componentwise. Thus  
 $\til P \RR \til Q$ means  $P_i \RR Q_i$
for all components.
Several  behavioural   relations
are used in this paper | \ref{a:tab}
presents  a summary of these.

\subsection{The call-by-value  \lc}
\label{ss:cbv}
We let $x$ and $y$ range over the set of $\lambda$-calculus  variables.
The set   $\Lao$  of $\lambda$-terms is  
 defined by the  grammar  
\begin{center}
$ M ::= \;   x    \midd   \lambda   x. M  \midd  
  M_1 M_2 \, .$ 
\end{center}
 Free variables, closed terms, substitution, 
$\alpha$-conversion
etc.\ are  defined as usual  \cite{barendregt1984lambda,DBLP:books/cu/HindleyS86}.
Here and in the rest of the paper (including when reasoning about
$\pi$ processes), we adopt the usual
``Barendregt convention''. This will allow us to assume freshness
of bound variables and names  whenever needed.
The set of   free variables 
in the term $M$  is written   $\fv M$, and we sometimes use $\fv{M,N}$
to denote $\fv M\cup \fv N$.
Application is left-associative;
therefore $M N L $  
is $(M N ) L$.
 We  abbreviate $\lambda  x_1. \cdots. \lambda  x_n.M $   as 
$\lambda  x_1 \cdots x_n.M $, or $\lambda  \tilde{x}. M$ if the length of
$\tilde x$ is not important.
Symbol $\Omega  $ stands for the  always-divergent term
 $(\lambda  x . x x)(\lambda  x . x x)$.



\txthere{explain the following things ("defined as usual"?):}{scope
  and associativity in the \lc, substitutions, free variables
  $\fv{M}$, $\fv{M,N}$ is short for $\fv M\cup \fv N$, values and
  evaluation contexts as syntactic categories (of terms and contexts),
  open terms (here, everything is defined for open terms)} 


A \emph{context} is a term with a hole \hdot,
possibly occurring more than once. If $C$ is a context, then $C[M]$ is
a shorthand for $C$ where the hole \hdot is substituted by $M$. An
 \emph{evaluation context}, ranged over using $\evctxt$, is a special kind of 
 inductively
 defined
context,
with exactly one hole \hdot, and in which a term replacing the hole can
immediately run. In the pure $\lambda$-calculus \emph{values} are abstractions and variables.
\begin{center}
\begin{tabular}{rl}
{Evaluation contexts} & $\evctxt\scdef \hdot\OR \evctxt M\OR\val \evctxt$ \\
{Values} & $\val\scdef x \OR \abs x M$
\end{tabular} 
\end{center}
 We
accordingly write \fv\evctxt{} for the free variables of \evctxt. 




Eager reduction (or $\betav$-reduction), 
  ${\red}
\subseteq    \Lao \times  \Lao $, 
is 
defined on open terms, and 
is determined  by the rule:
\txthere{}{Eager reduction relation (definition?):}
$$  \evctxt[(\lambda x . M) \val]\red \evctxt[M\{\val/x\}]
,
$$
where $\{\val/x\}$ stands for the capture-avoiding substitution of $x$
with $\val$. 

A term in \emph{eager normal form} is a term that has no eager
reduction.
We write \reds\ for the reflexive transitive closure of \red.

\begin{proposition}\label{p:case:analysis}
The following hold:
\begin{enumerate}
\item If $M\red M'$, then $\evctxt[M]\red \evctxt[M']$ and $M\sigma
  \red M'\sigma$,
  for any  substitution $\sigma$ that replaces variables with values. 
\item Terms in eager normal form are either values or admit 
a unique decomposition of the shape 
  $\evctxt[x \val]$.
\end{enumerate}
\end{proposition}

Therefore, given a term $M$, either $M\reds M'$ where $M'$ is a term in
eager normal form, or there is an infinite reduction sequence 
starting from $M$.
In the first case,  we say that $M$ \emph{has eager
normal form $M'$}, written $M\converges M'$; in the second $M$
\emph{diverges}, written $M\diverges$. We write $M\converges$ when
$M\converges M'$ for some $M'$.


\begin{definition}[Contextual equivalence]\label{d:ctxeq}
Given $M,~N\in \Lao$, we say that $M$ and $N$ are contextually
equivalent, written $M\ctxeq N$, if for any context $C$, we have 
$C[M]\converges$ iff $C[N]\converges$. 
\end{definition}


\subsection{Tree Semantics for call-by-value}
\label{ss:trees}
 In this section, 
we recall 
Lassen's 
\emph{eager normal-form bisimilarity} \cite{lassentrees,lassentrees2,lassenfa}.

\begin{definition}[Eager normal-form bisimulation] 
\label{enfbsim}
A relation $\R$ between \lterms is an \emph{eager normal-form bisimulation} if, whenever
$M\RR N$, one of the following holds: 
\begin{enumerate}
\item both $M$ and $N$ diverge;
\item 
\label{ie:split}
$M\converges \evctxt[x\val]$ and $N\converges \evctxtp[x\valp]$ for
some $x$, values $\val$, $\valp$, 
and evaluation contexts 
$\evctxt$ and $\evctxtp$
satisfying 
$\val\RR
\valp$ and $\evctxt[z]\RR \evctxtp[z]$ for a  fresh $z$;
\item $M\converges \abs x M'$ and $N\converges \abs x N'$ for some $x$, $M'$, $N'$ with $M'\RR N'$;
\item $M\converges x$ and $N\converges x$ for some $x$.
\end{enumerate}
\emph{Eager normal-form bisimilarity}, written $\enf$, is the largest eager normal-form bisimulation.
\end{definition}


Essentially, the structure of a $\lambda$-term that is unveiled by
Definition~\ref{enfbsim} is that of a (possibly infinite) tree 
obtained by repeatedly applying $\betav$-reduction, and branching a tree whenever
instantiation of a variable is needed to continue the reduction (clause \reff{ie:split}).  
We call such trees \emph{Eager Trees} (ETs); accordingly, we also call  
 eager normal-form bisimilarity the \emph{Eager-Tree equality}.


\begin{example}
\label{exa:cteq}
Relation $\enf$ is strictly finer than contextual equivalence
$\ctxeq$: the inclusion ${\enf} \subseteq {\ctxeq}$ follows from the
congruence properties of $\enf$ \cite{lassentrees}. For strictness,
examples are given by the following equalities, which hold for
$ \ctxeq$ but not for $\enf$:
  $$
  \Omega  = (\abs y \Omega) (x\val)
\qquad\qquad
x\val =  (\abs y x\val)(x\val)
\enspace.
$$
\end{example} 
\begin{example}[$\eta$ rule]
\label{exa:eta}
The  $\eta$-rule is not valid for $\enf$. For instance, we have 
$\Omega \not\enf \abs x \Omega x$.
The rule is not even valid on values, as we also have
$ \abs y x y \not\enf x$. It holds however 
for abstractions: 
$   \abs y (\abs xM) ~y\enf \abs xM
$ when $y\notin\fv{M}$.
\end{example} 
The failure of the   $\eta$-rule $ \abs y x y \not\enf x$ is
troublesome as, under any closed value substitution (a substitution
replacing variables with closed values), the two terms are
indeed {
eager normal-form bisimilar}. 
%
%
%
%
Thus \emph{$\eta$-eager normal-form bisimilarity}~\cite{lassentrees} takes
$\eta$-expansion into account so to recover such missing equalities.


\begin{definition}[\enfbsim] \label{enfebsim}
A relation $\R$ between \lterms is an \emph{\enfbsim} if, whenever $M\RR N$, either one of the clauses of Definition \ref{enfbsim},
or one of the two following additional clauses,  hold:
\begin{enumerate}
\setcounter{enumi}{4}
\item\label{lab:five} $M\converges x$ and $N\converges \abs y N'$ for some
 $x$, $y$, and $N'$ such that $ N'\converges\evctxt[x\val]$,
with      $y\RR \val$ and $z\RR
  \evctxt[z]$ for some value $\val$, evaluation context 
\evctxt, and  fresh $z$.
\item\label{def:enfe:case:eta}
the converse of \reff{lab:five}, i.e., 
$N\converges x$ and
  $M\converges \abs y M'$ for some
  $x$, $y$, and $M'$ such that $ M'\converges\evctxt[x\val]$,
with
     $ \val\RR y$ and 
$
   \evctxt[z]\RR z$ for some value $\val$, evaluation context 
 \evctxt, and  fresh $z$.
\end{enumerate}
Then \emph{$\eta$-eager normal-form bisimilarity}, $\enfe$, is the largest $\eta$-eager normal-form bisimulation.
\end{definition}
We sometimes call relation $\enfe$ the \emph{$\eta$-Eager-Tree equality}. 
\begin{remark}Definition~\ref{enfebsim} coinductively allows
$\eta$-expansions to occur underneath other $\eta$-expan{-}sions, hence 
trees with infinite $\eta$-expansions may be equated with finite trees.
For instance, we have
$$x\enfe \abs y xy\enfe \abs y x (\abs z yz)\enfe \abs y x (\abs z y(\abs w zw))\enfe\dots$$
An example of a finite tree being equated with an infinite tree by
$\enfe$
is
as follows: take a fixpoint combinator $Y$, and define 
$f \defi (\abs {zxy} x(z\,y))$. We then have $Yfx \reds \abs y x (Y\,f\,y)$,
and then $x\,(Y\,f\,y)\reds x(\abs z y\,(Y\,f\,z))$, and so on. Hence,
we have  $x\enfe Yfx$.
\end{remark}



\subsection{The \pc}
\label{s:pi}


In all  encodings
we consider, 
  the encoding of a  $\lambda$-term
 is parametric on  a name, 
i.e., it is a 
 function from names  to $\pi$-calculus  
 processes. We also need parametric processes (over one or several names) to write recursive process definitions
 and equations. 
We call such  parametric processes {\em abstractions}.
The instantiation of the parameters of an abstraction $F$ is done
via the {\em application} construct $\app F \tila$.
 We use $P,Q$ for processes, $F$ for abstractions.
Processes and abstractions  form the set  of  {\em $\pi$-agents} (or
simply \emph{agents}), ranged
over by $A$. 
Small letters 
 $a,b, \ldots, x,y, \ldots$  
range  over the infinite set of names.
The grammar of the $\pi$-calculus is thus:
$$ \begin{array}{ccll}
A & ::= & P \midd F & \mbox{(agents)}\\[\mypt]
P & ::= & \nil    \midd    \inp a \tilb . P    \midd    \out a \tilb . P 
   \midd    \res a P
& \mbox{(processes)}  \\[\myptSmall]
& &
   \midd     P_1 |  P_2   \midd  ! \inp a \tilb . P   
\midd \app F \tila
\\[\mypt]
F & ::= & \bind \tila P \midd K & \mbox{(abstractions)}
   \end{array}
 $$

$\nil$ is the inactive process. An input-prefixed process $a (\tilb) .
P$, where  $\tilb$ has  pairwise distinct components,
 waits for a tuple of  names $\tilc$
 to be sent along $a$ and then  behaves like
$P \sub {\tilc} \tilb $, where $\sub{\tilc}\tilb $ is the 
 simultaneous
substitution
of names 
$\tilb$ with names  $\tilc$ (see below). An output particle 
$\opw a \tilb  $  emits  names  $\tilb$ at $a$.
Parallel composition  is used  to run two processes in parallel.
The restriction $\res a P$ makes name $a$ local, or private, to $P$.
A replicated input  $ !\inp a\tilb.P$
stands for a countable infinite  number
of copies of $\inp a\tilb.P$ in parallel.
 (Replication could be avoided in the syntax since it can be encoded
 with recursion. However its semantics is simple, and it is a useful
 construct for examples and encodings; thus we chose to include it in
 the grammar.)
 %

We do not include the operators of sum and matching.
We assign parallel composition the lowest precedence among the
operators. 
We refer to~\cite{milner1993polyadic} for detailed discussions on the
operators of the language.

 We use $\alpha $
 to range over prefixes.
 In  prefixes $\inp a \tilb$ and $\out a \tilb$, we call 
 $a$ the {\em subject} and  $\tilb$  the {\em object}. 
When the tilde  is empty, the surrounding brackets in prefixes
 will be omitted.
We  
 often abbreviate 
 $\alpha  . \nil$ as $\alpha $, and 
$\res a \res b P$ as $\resb{
a,b} P$.
An input prefix  $a (\tilb) .P$, a restriction 
   $\res{b} P$, and an abstraction $\bind\tilb P$  are binders for 
   names  $\tilb$ and $b$, respectively, and  give
rise in the expected way  to the definition of {\em free names}
(\mbox{\rmsf fn}) and {\em
bound names} (\mbox{\rmsf bn})
 of a term or a prefix. 
An agent is \emph{name-closed}  if it does not contain free names. 
 (Since the
 number of recursive definitions may be infinite, some care is
 necessary in the definition of free names of an agent, using a least
 fixed-point construction.)
As in the $\lambda$-calculus, 
we   identify
processes or actions which only  differ in the choice of the 
 bound names. 
The symbol $=$  will 
mean  ``syntactic identity modulo
$\alpha$-conversion''.
Sometimes, we  use $\defi$   as abbreviation mechanism, to
assign a name to an   expression to which we want to refer later.

   We use constants, ranged over by $K$, for writing recursive definitions. Each
constant has a defining equation of the form 
$K \Defi  \bind{\tilx} P$, where $\bind{\tilx} P  $ is name-closed; $\tilx$ 
 are the formal parameters of
the constant (replaced by the actual parameters whenever the constant
is used).



Since the calculus is polyadic, 
we assume a \emph{sorting system}~\cite{milner1993polyadic}
   to avoid disagreements  in the arities of
the tuples of names 
carried by a given name 
and in applications of abstractions.
We do not present the sorting system 
because it is not essential. 
The reader should
take for granted that all agents described  obey  a sorting. 

A \emph{context}  $\qct$ of the $\pi$-calculus    is a $\pi$-agent in which some
subterms have been replaced by the hole $\hdot{}$ or, if the context is
polyadic, with indexed holes $\holei 1, \ldots, \holei n$. Then, 
 $\ct A$ or $\ct {\til A}$
 is the agent resulting from replacing the holes with the terms $A$ or
 $\til A$.
 Holes in
contexts have a sort too, as they could be in place of an abstraction.
 




Substitutions are of the form $\sub \tilb \tila$, and
are   finite assignments of  names to names.
We use $\sigma$ and $\rho$  to range over substitutions. 
The application of a
  substitution  $\sigma $ to an expression $H$ is
written $H \sigma $.
Substitutions have 
precedence over the operators of the language; $\sigma  
\rho$ is the composition of substitutions  where $\sigma $ is performed
first, 
 therefore $P \sigma  
\rho$ is $ (P \sigma ) \rho$.

 The Barendregt convention allows us to assume that 
the application of a substitution
does not   affect  
   bound names   of expressions;
similarly, when comparing the transitions of two processes, we 
assume that the bound names of the actions 
do not occur free in
the processes.
In  a statement, we say that  a name is {\em fresh} to mean that it
 is different from any other name which 
occurs  in the statement or in objects of the statement like
processes and substitutions.

 \paragraph{Abstraction and application}
We say that an {application redex} $((\til x) P)\param{\til a}$ can be \emph{normalised} as $P \subst {\til x}{\til a}$. An agent is
\emph{normalised} if all such application redexes have been
contracted, everywhere in the terms. 
When reasoning on behaviours
it is useful to assume that all expressions are normalised, in the above sense.    Thus
in the remainder of the paper \emph{we identify an agent with its normalised expression}.
The application construct $ F \param{\til a}$ will play an important role in the treatment of
equations in the following sections.

\subsection{Operational semantics}

Transitions of $\pi$-calculus processes are of the form $P\arr\mu P'$,
where the grammar for actions is given by
$$
\mu\quad::=\quad \inp a\tilb\midd \res{\til{d}}\out a\tilb\midd \tau
\enspace.
$$
\begin{itemize}
\item $ P \arr{\inp a \tilb}P'$ is an input, where $\tilb$
 are the names bound by the input prefix which is being fired (we adopt a late version of the Labelled Transition Semantics),
 \item $P
 \arr{\res {\til{d}}\out a\tilb}P'$ is an output, where $\til d \subseteq
 \tilb$ are private names extruded in the output, and 
 \item $P \arr\tau P'$
 is an internal action.
\end{itemize}
%
%
%
We abbreviate
$\res{\tild}\out a\tilb$ as $\out a\tilb$ when $\tild$ is empty.
The occurrences 
of $\tilb$ in $\inp a\tilb$ and those of 
$\til{d}$ in $
\res{\til{d}}\out a\tilb$ are bound; we define accordingly 
the sets of bound 
names and free names 
 of an action $\mu$, respectively
written $\bnames\mu$
and $\fnames\mu$. 
The set of all the names appearing in $\mu$ (both free and bound)
is written $\names\mu$. 

Figure~\ref{f:lts:pi} presents the transition rules for the \pc.

\begin{figure*}[t]
\begin{mathpar}
  \inferrule{~}{\inp a\tilb.P \arr{\inp a\tilb}P}
  \and
  \inferrule{~}{!\inp a\tilb.P \arr{\inp a\tilb}!\inp a\tilb.P | P}
  \and
  \inferrule{~}{\out a\tilb.P\arr{\out a\tilb}P}
  \and
  \inferrule{P\arr{\res{\til{d}}\out a\tilb}P'
  }{
    \res n P\arr{\res{(\{n\}\cup\til{d})}\out a\tilb} P'}
  \begin{array}{l}
  n\in\tilb\\ n\notin\tild
  \end{array}
  \and
  \inferrule{P\arr\mu P'
  }{
    \res n P\arr\mu \res n P'}
  n\notin\names{\mu}
  \and
  \inferrule{P\arr{\inp a \tilb}P'\and Q\arr{\res{\til{d}} \out a{\til{b'}}}Q'
  }{
    P | Q\arr\tau \res{\til{d}}(P'\sub{\til{b'}}{\tilb} | Q')
  }
  \and
  \inferrule{P\arr\mu P'}{P | Q\arr\mu P' | Q}\bnames\mu\cap\fnames
  Q=\emptyset
  \and
  \inferrule{P\sub{\tilb}{\tila}\arr\mu P'}{\app {(\bind\tila P)}\tilb \arr{\mu} P'}
  \and
  \inferrule{\app F\tila\arr\mu P'}{\app K\tila\arr\mu P'} \mbox{ if }K\Defi F
\end{mathpar}
  \caption{Labelled Transition Semantics for the \pc}
  \label{f:lts:pi}
\end{figure*}

 We write
 $\Arr {}$ 
 for the reflexive transitive closure of $\arr{\tau}$, and 
 $\Arr{\mu}$ for $\Longrightarrow   \arr{\mu}\Longrightarrow$. Then
 $\Arcap \mu$ (resp.\ $\arr{\hat\mu}$) is $\Arr{\mu}$ (resp.\ $\arr\mu$) if $\mu$ is not
 $\tau$, and $\Arr{}$ (resp.\ $\arr\tau$ or $=$)
 otherwise.
 In bisimilarity  or other behavioural relations for the
 $\pi$-calculus we consider, no name instantiation
 is used  in the input clause or elsewhere; technically, the relations are \emph{ground}. 
 In the subcalculi we consider ground bisimilarity is a congruence and coincides with
 barbed congruence (congruence breaks in the full $\pi$-calculus).  Besides the simplicity
 of their definition, the ground relations make more effective the theory of unique
 solutions of equations (in particular, checking divergences is simpler, see Section~\ref{s:usol}).

The reference behavioural equivalence for $\pi$-calculi is
the usual \emph{barbed congruence}.
We recall its definition,  on a generic subset $\LL$ of
$\pi$-calculus processes.
 A \emph{$\LL$-context} is a  process of $\LL$ with a single
hole $\contexthole$ in it (the hole has a  sort, as it could be in
place of an abstraction). 
We write $P \Dwa_{a}\;$  
if $P$ can make an output action
whose subject is $a$, possibly after some internal moves. 

 We make only
output observable because this is standard in asynchronous
calculi; in the case of a synchronous calculus like \Intp, Definition~\ref{d:bc}
below yields synchronous barbed congruence, and adding also
observability of inputs does not change the induced equivalence. 
More details on this are given in Section~\ref{s:localpi}.

\begin{definition}[Barbed  congruence]
\label{d:bc}
 {\em Barbed bisimilarity} is the largest 
symmetric relation $ \wbb$ on 
 $\pi$-calculus  processes   
such that
  $P \wbb  Q$ implies:
\begin{enumerate}
\item 
If $P \Longrightarrow P'$ then  there is $ Q'$ such that
$Q \Longrightarrow Q'$ and
      $P' \wbb Q'$.
\item  $P \Dwa_{ a}\/$  iff $Q \Dwa_{a}\/$.
\end{enumerate}
Let $\LL$ be a set of $\pi$-calculus agents, and 
 $A, B \in \LL$. We say that $A$ and $B$ are 
  {\em barbed congruent in 
$\LL$}, written $A \wbc\LL B$,
if for each (well-sorted)  $\LL$-context $C$, it holds that $C[A] \wbb C[B]$.
\end{definition}

\begin{remark}
\label{r:abs} 
We have defined barbed congruence uniformly on processes and
abstractions (via a quantification on all process contexts). 
Usually, however,  
 definitions will only be given for processes; it is
then intended that they are extended to abstractions by requiring
closure under ground parameters, i.e., by supplying fresh names as
arguments. 
\end{remark} 

As for
 all contextually-defined behavioural relations, so 
barbed congruence is
hard to work with.  In all calculi we consider, it can be
characterised in terms of \emph{ground bisimilarity}, under the (mild)
condition that the processes are image-finite up to weak bisimilarity.
{
 (We recall  that
the class of processes {\em image-finite up to weak bisimilarity} 
is the largest subset ${\mathcal {IF}}$ of
$\pi$-calculus processes  which is  closed by transitions  and such that 
$P \in {\mathcal {IF}}$ implies that,   for all actions $\mu$, 
the set 
$\{P'  \st    P \Arr{\mu } P'\}$
quotiented by weak bisimilarity 
is finite. The definition is extended to abstractions as by
Remark~\ref{r:abs}.) }
All the agents in the paper, including those obtained by encodings
of the $\lambda$-calculus, are image-finite up to weak bisimilarity. 
The distinctive feature of \emph{ground} bisimilarity is that it does not
involve instantiation of the bound names of inputs (other than by
means of fresh names), and similarly for abstractions.
In the remainder, we omit the adjective `{ground}'. 

\begin{definition}[Bisimilarity]
\label{d:bisimulation}
A symmetric relation $\R$ 
on $\pi$-processes is a
\emph{bisimulation}, if whenever $P \,\R\, Q$ and $P \arr\mu P'$, then $Q \Arcap\mu Q'$
for some $Q'$ with $P' \,\R\, Q'$. 

Processes $P$ and $Q$ are \emph{bisimilar}, written
$P\approx Q$, if $P \,\R\, Q$ for some bisimulation $\R$. 
\end{definition}

 We extend $\approx$ to abstractions, as per Remark~\ref{r:abs}: 
$F 
\approx G 
$ if $\app {F}\tilb \approx \app {G} \tilb $ for fresh $\tilb$.

In the proofs, we shall also use \emph{strong bisimilarity}, written
$\sim$. Relation $\sim$ is defined as per Definition~\ref{d:bisimulation}, but
$Q$ must answer with a strong transition, that is, we impose $Q\arr\mu Q'$.

\paragraph{The Expansion preorder}
             
We define the expansion preorder, written $\lexn$, where $P\lexn Q$
intuitively means that $P$ and $Q$ have the same behaviour, and that
 $P$ may  not be `slower' (in the sense of doing more $\arr\tau$
transitions) than process $Q$.

\begin{definition}
\label{d:expan}
  \emph{Expansion}, written $\lexn$, is defined as the largest relation $\R$
  such
  that $P\RR Q$ implies
  \begin{enumerate}
  \item
 if $P\arr\mu P'$, 
then for some $Q'$,
  $Q\Arr{\mu}Q'$ and $P'\RR Q'$, and
\item if $Q\arr\mu Q'$, 
then
  for some $P'$, $P\arr{\hat\mu}P'$ and $P'\RR Q'$. 
  \end{enumerate}

  The converse of $\lexn$ is written $\exn$.
\end{definition}
As usual, expansion is extended to abstractions by requiring ground
instantiation of the parameters: $F\lexn F'$ if
$\app{F}{\tila} \bsim \app{F'}{\tila}$, where $\tila $ 
are
fresh names of the appropriate sort.

In the following we shall use two standard properties of expansion:
first, expansion is finer than bisimilarity, i.e., 
${\lexn}\subseteq{\bsim}$. Second, expansion, like bisimilarity, is
preserved by all contexts.




\subsection{The Subcalculi \Intp~and \alpi}\label{s:subcalculi}

We focus on two subcalculi of the $\pi$-calculus: the Internal $\pi$-calculus (\Intp),
and the Asynchronous Local $\pi$-calculus (\alpi). They are obtained by placing certain constraints
on prefixes. 

\paragraph{\Intp}

In \Intp, all outputs are bound. This is 
syntactically enforced  by replacing the output construct with 
the bound-output construct $\bout a \tilb .P$, which, with respect to
 the grammar
of the ordinary  $\pi$-calculus, is an abbreviation for $\res \tilb
\out a \tilb . P$. In all tuples (input, output, abstractions, applications) the
components are pairwise distinct so to make sure that distinctions among names are
preserved by reduction. 



\txthere{Forwarders}{How to fix forwarders for internal pi: infinite forwarders. We need also:\\
We consider a sorting system over the channel names, so that a name carries its sort (and therefore the number of names that are transmitted over it, as well as whether the channel is linear or not).
}

\begin{theorem}
\label{t:bisbc}~
 In  \Intp, on agents that are image-finite up to~\bsim,  barbed congruence  and
bisimilarity 
coincide. 
 \end{theorem} 

The encoding of the $\lambda$-calculus into  \Intp\ 
yields processes that are image-finite up to \bsim. 
Thus we can use 
bisimilarity as a proof technique for barbed congruence. 
 
\paragraph{\alpi}

\alpi\ is defined by enforcing that in an input $\inp a\tilb.P$, all
names in $\tilb$ appear only in output position in $P$.
Moreover, 
 \alpi{} being  \emph{asynchronous},
  output prefixes have no continuation; 
in the grammar of the $\pi$-calculus this corresponds to having only outputs of the form 
 $\out a\tilb.\nil$ (which we will simply write $\out a\tilb$).

 In   
\alpi, to maintain the  characterisation of  barbed congruence as (ground) bisimilarity,
the transition system has to be modified ~\cite{localpi}, 
 allowing the dynamic introduction of additional processes (the
  `links', sometimes also
  called forwarders).
In Section~\ref{s:localpi}, we present the modified transition system for
\alpi, upon which weak bisimilarity is defined.
We also explain how they allow us to obtain for \alpi{} a property  
similar to that of Theorem~\ref{t:bisbc} for \Intp. 

\section{Unique solutions in  \Intp\ and \alpi}
\label{s:usol}

We adapt   the proof technique of unique solution of equations, from~\cite{usol} to
the calculi  \Intp\ and \alpi, in order to  derive bisimilarity results. 
The technique is discussed in ~\cite{usol} on the asynchronous $\pi$-calculus
(for possibly-infinite systems of equations).
The structure of the 
 proofs for \Intp\ and \alpi\ is similar; 
 in particular the completeness part 
is essentially the same 
 because
 bisimilarity  is the same. 
 The differences in the syntax of \Intp, and in the transition
 system of \alpi, show up only in certain technical details of the
 soundness 
 proofs. 

%
\iffull
The results presented in this section hold both for \Intp\ and for
\alpi. 
\fi




\paragraph*{Equation expressions}
We need variables to write equations. We  use
 capital
letters  $X,Y,Z$
 for  these variables and call them \emph{equation variables}
 (sometimes simply \emph{variables}).
 The body of an equation is a name-closed abstraction
possibly containing equation  variables 
(that is,  applications can also be of the form $\app X\tila$). 
Thus, the solutions of equations  are abstractions. 
Free names of
equation expressions and contexts are defined as for agents. 
%

%
We use $E$ to range over name-closed abstractions; and 
 $\EE$ to range over systems of equations, defined as follows.
 In all the definitions, the indexing set $I$ can be infinite.

\begin{definition}
Assume that, for each $i$ belonging to 
 a countable indexing set $I$, we have a variable $X_i$, and an expression
$E_i$, possibly containing  some variables. 
Then 
$\{  X_i = E_i\}_{i\in I}$
(sometimes written $\til X = \til E$)
is 
  a \emph{system of equations}. (There is one equation for each
  variable $X_i$; we sometimes use $X_i$ to refer to that equation.)

  A system of equations is \emph{guarded} if each
  occurrence of a variable in the body of an equation is underneath a
  prefix.
\end{definition}


$E[\til F]$  is the abstraction resulting from $E$ by
replacing each occurrence of the variable $X_i$ 
with the abstraction $F_i$ 
(as usual assuming
$\til F$ and $\til X$ have the same sort). 
  This is syntactic
  replacement.  
However recall that we identify an agent with its normalised expression: hence 
replacing $X$ with $(\til x)P$ in $X\param{\til a}$ is the same as replacing 
$X\param{\til a}$ with the process $P\subst{\til x}{\til a}$.


\begin{definition}\label{d:un_sol}
Suppose  $\{  X_i = E_i\}_{i\in I}$ is a system of equations. We say that:
\begin{itemize}
\item
 $\til F$ is a \emph{solution of the 
system of equations  for $\bsim$} 
if for each $i$ it holds
that $F_i \bsim E_i [\til F]$.
\item The  system has 
\emph{a unique solution for $\bsim$}  if whenever 
$\til F$ and $\til G$ are both solutions for $\bsim$, we have 
$\til F \bsim \til G$. \end{itemize}
 \end{definition} 

 \begin{definition}[Syntactic solutions]
   The \emph{syntactic solutions} of the system of equations $\til X =\EeqBody{}{}$ are the  
   recursively defined constants $\KEi E \Defi E_i[\KE]$, for each
   $i\in I$, where $I$ is the indexing set of the system of equations.
 \end{definition}
The syntactic solutions
of a system of equations 
are indeed solutions of
it. 
%




  A process $P$ \emph{diverges} if it can perform an infinite sequence
  of internal moves, possibly after some visible ones 
  (i.e., actions  different from $\tau$); formally,  there are
  processes $P_i$, $i\geq 0$, and some $n$, such that
  $P=P_0\arr{\mu_0} P_1 \arr{\mu_1} P_2 \arr{\mu_2}\dots$ and for all
  $i>n$, $\mu_i=\tau$. We call a \emph{divergence of $P$} the sequence
  of transitions $\big(P_i\arr{\mu_i}P_{i+1}\big)_{i}$.
 In the case of an abstraction $F$, as per Remark~\ref{r:abs},
$F$ has a divergence 
if the process $\app F\tila$ has a divergence, where $\tila$ are fresh
names. A tuple of agents $\til A$ \emph{is divergence-free} if none of the
components $A_i$ has a  divergence. 

The following result is the technique we rely on to establish
completeness of the encoding. As announced above, it holds both in
\Intp\ and in \alpi.
\begin{theorem}\label{thm:usol}
In \Intp\ and \alpi,
a guarded system of equations 
with divergence-free  syntactic 
solutions  has unique solution for \bsim. 
\end{theorem}

The proof of Theorem~\ref{thm:usol} is very similar to a similar
theorem, for the \pc, which is presented in~\cite{usol}.
%
Moreover, we present a proof of unique solution for trace inclusion and trace equivalence (in \Intp)
 in~\ref{a:usoltrace}. The latter is also very similar in structure
to the two aforementioned proofs.
For a pedagogical presentation of these proofs,
we refer the reader to \cite{usol}, specifically to 
the proof of unique solution for weak bisimilarity in the setting of 
CCS.


Techniques for ensuring termination, hence divergence freedom, have been 
studied in, e.g., \cite{termination1,termination2,termination3}.

\subsection*{Further Developments of the Technique of Unique Solution
  of Equations.}


We  present some further developments to the theory of unique solution of equations, 
that are needed  for the results in this paper. 
The first result allows us to derive the unique-solution property for a system of
equations from the analogous property of an extended system.

\begin{definition}
\label{d:extend}
A system of equations $\Eeq'$ \emph{extends} system $\Eeq$
if there exists a fixed set of indices $J$ such that any solution of
$\Eeq$ can be obtained from a solution of $\Eeq'$ by removing the
components corresponding to indices in $J$. 
\end{definition} 


\begin{lemma}
  \label{t:transf:equations}
  Consider two systems of equations 
$\Eeq'$ 
and $\Eeq$ 
where 
  $\Eeq'$
extends $\Eeq$. 
If $\Eeq'$ has a unique solution, then the property also holds for
$\Eeq$. 
\end{lemma}

We shall use Lemma~\ref{t:transf:equations} in
Section~\ref{s:complete}, in a situation where we transform a certain
system into another one, whose uniqueness of solutions
is easier to establish. 
\iffull
Then, by 
Lemma~\ref{t:transf:equations}, the property 
holds for the initial system.




\fi

\begin{remark}
We cannot derive Lemma~\ref{t:transf:equations} by comparing the syntactic solutions of
the two systems  $\Eeq'$ and $\Eeq$.
For instance,  the equations
  $X=\tau.X$ and $X=\tau.\tau\dots$ (where $\tau.\tau\dots$ abbreviates
  the corresponding recursive definition)  have strongly bisimilar   syntactic 
  solutions, yet only the latter equation has the unique-solution
  property.
(Further, Lemma~\ref{t:transf:equations} allows us to compare systems
of different size). 
\end{remark}

\medskip

The second development is a generalisation of 
Theorem~\ref{thm:usol} and Lemma~\ref{t:transf:equations} to trace equivalence and 
to trace preorders; we postpone its presentation to Section~\ref{s:contextual}.




\section{Milner's Encodings}
\label{s:enc:cbv}
\subsection{Background: encodings $\qencm$ and $\qencmp$}


Milner noticed \cite{milner:inria-00075405,encodingsmilner} that his
call-by-value 
 encoding can be easily tuned so to mimic
forms of 
evaluation in which, in an application $MN$, the function $M$ is run
first, or the argument $N$ is run first, or function and argument are
run in parallel (the proofs are actually carried out for this last
option). We chose here the first option.
%

The core of any  encoding of the  $\lambda$-calculus  into a process calculus is the 
translation of function application. This  
becomes a particular form of 
parallel combination of two processes, the function and its argument;  
 $\betav$-reduction  is then  modelled as   process interaction.

The encoding of a $\lambda$-term is parametric over a name; this may
be thought of as the  \emph{location} of that term, or as its
\emph{continuation}.  
A term that    becomes a value   signals so at its continuation name
and, in doing so, 
it grants 
access 
 to the body of the
value. Such body   is replicated, so that the value may be
copied several times. When the value is a function, its body can
receive two names: (the access to) its value argument, and the
following continuation.
In the translation of application, first the function is run, then
the argument; finally the function is informed of its argument and 
continuation.

In the  original paper~\cite{milner:inria-00075405},
Milner presented
two candidates for the encoding of call-by-value
$\lambda$-calculus~\cite{DBLP:journals/tcs/Plotkin75}.  They follow
the same pattern of 
translation, in particular regarding application (as described above),
but with a technical difference in the rule for
variables.  One encoding, $\qencm$, is defined as follows
 (for the case of application, we adapt the encoding from parallel call-by-value to
left-to-right call by value, as described above):
\begin{align*}
\encm{\abs xM} & \deff 
\bind p\outb p y.!\inp y{x,q}.\encma Mq \\
\encm{MN}& \deff  \bind p(\new q)(\encma Mq
|\inp qy.\new r(\encma Nr| \inp rw.\out
  y{w,p})) \\
\encm{x}& \deff \bind p \out p x
\end{align*}
In the other encoding, $\qencmp$,
application and $\lambda$-abstraction are treated as in 
$\qencm$; the rule for variables is:
$$
\encmp{x}  \deff 
\bind p \outb p y.!\inp y {z,q}.\out x{z,q}
\enspace.
  $$



In $\qencmp$, a
 $\lambda$-calculus variable gives rise to a one-place buffer.  As the
computation proceeds, these buffers are chained together, gradually
increasing the number of steps necessary to simulate a
$\beta$-reduction.  This phenomenon does not occur in $\qencm$, where
the occurrence of a variable disappears after it is used.
Hence the encoding $\qencm$ is  more efficient than $\qencmp$,

\subsection{Some problems with the encodings}
\label{ss:ot}

The immediate free output in the encoding of variables in  $\qencm$ 
breaks the validity of $\betav$-reduction; i.e., there exist a term $M$
and a value $V$ such that $\encm{(\lambda x. M)V } \not\bsim
\encm{ M \sub V x }$~\cite{sangiorgiphd}. 
The encoding $\qencmp$ fixes 
this by communicating, instead 
of a
free name,  a 
fresh  pointer to that name. 
 Technically, the initial free output of $x$ is replaced by a
 bound output coupled with a link to $x$ (the process
 $!\inp y {z,q}.\out x{z,q}$, receiving at $y$ and re-emitting at $x$).
 Thus  $\betav$-reduction is validated, i.e., 
$\encm{(\lambda x. M)V } \bsim
\encm{ M \sub V x }$ for any  $M$ and $V$~\cite{sangiorgiphd}.

%
(The final version of Milner's paper~\cite{encodingsmilner} 
was
written after the results in~\cite{sangiorgiphd} were known and presents
only the encoding $\qencmp$.)

Nevertheless,      $\qencmp$ only delays the free output, as the added link
contains itself a free output. 
As a consequence, we can show that other desirable equalities of
call-by-value are broken in
$\qencmp$. An example is law~\reff{eq:nonlaw} from the
Introduction, as stated by Proposition~\ref{p:nonlaw} below.
 This law is desirable (and indeed valid for contextual equivalence,
 or the Eager-Tree equality) 
intuitively because, in any substitution closure 
of the law,  either both terms diverge, or they
 converge to  the same value. 
 The same argument holds 
 for their $\lambda$-closures, $\abs x x\val$ and $\abs x I(x\val)$. 

 We recall that $\wbc\pi$ is barbed congruence in the 
$\pi$-calculus.
\begin{restatable}[Non-law]{proposition}{nonlaw}\label{p:nonlaw}
For any value $\val$, we have:
$$
\encm {I(x\val)}
\nwbc\pi
 \encm {x\val}
\quad\mbox{ and }\quad
\encmp {I(x\val)} 
\nwbc\pi
\encmp {x\val} 
\enspace.
$$
\end{restatable}
(The law is
violated also under coarser equivalences, such as
contextual  equivalence.) 
Technically, the reason why the law fails in $\pi$ can be illustrated
when $\val=y$, for encoding $\qencm$. We have:
\begin{alignat*}{3}
  \encma {xy} p &\wbc\pi \out x {y,p}
                    \mbox{ \hspace{1ex}}
                    \\[\mypt]
  \encma {I(xy)} p&\wbc\pi (\new q)( \out x {y,q} | \inp q z.\out p z  )    
\end{alignat*}

(details of the calculation may be found in~\ref{a:milner})

In presence
of the normal form $x y$, the identity $I$ becomes observable. Indeed, in the second
term, a fresh name, $q$, is sent instead of continuation $p$, and a
link between $q$ and $p$ is installed. This corresponds to a
law
which is valid in 
\alpi, but not in $\pi$.


  \begin{remark}[Generalisations of law~\reff{eq:nonlaw}]
    The phenomenon observed in law~\ref{eq:nonlaw} can be observed in
    a more general setting.
    We can observe that for any evaluation context $\evctxt$ and any value $\val$, we have 
$$\encma {\evctxt[I (x \val)]}p \not\wbc\pi \encma {\evctxt[x
\val]}p
\enspace.$$
One may want to generalise further this law, by replacing the identity $I$ by an 
arbitrary function $\abs z M$, provided $M$ somehow ``uses'' $z$. 
We may take $M$ to be equal to $\evctxt' [z]$, for some evaluation context $\evctxt'$, 
or any term that reduces to 
such a normal form. 
We can then show, indeed: 
$$\encma {\evctxt[\abs zM (x \val)]}p \not\wbc\pi \encma {\evctxt[M\subst z{x
\val}]}p
\enspace.$$
Generalising further, to a term $M$ whose normal form is not written $\evctxt'[z]$, 
is outside the scope of this work. This could be related to Accattoli
and Sacerdoti Coen's notion of ``useful reduction'' for the call-by-value \lc~\cite{useful}.
\end{remark}


\subsection{Well-behaved encodings}


The problem put forward in Proposition~\ref{p:nonlaw} can be avoided by iterating the transformation that takes us
from $\qencm $ to $\qencmp$ (i.e., the replacement of a free output
with a bound output so to avoid all emissions of free names). Thus the
target language becomes Internal $\pi$; the resulting encoding is
analysed in Section~\ref{s:enc:pii}.

Another solution is to control the use of name capabilities in
processes. In this case the target language becomes \alpi,
 and we need not modify
 the  initial encoding  $\qencm$. This
situation is analysed in Section~\ref{s:localpi}.


In both solutions, the encoding uses two kinds of names: \emph{value names} $x,y,
\dots, v,w,\dots$ and  \emph{continuation names} $p,q,r,\dots$,.
For  simplicity,  we assume that  the set of value names contains
the set of $\lambda$-variables. 

Continuation names are always used linearly, meaning that they are 
only used once in subject position (once in input and once in output). On the other hand,  value 
names may be used multiple times.
Continuation names are used to transmit a value name, and value names
are used to transmit a pair consisting of a value name and a
continuation name.
This is a very mild form of typing. We could
avoid the distinction between these two kinds of names, at the cost of
introducing additional replications in the encoding. 

Moreover, in both solutions,
the use of  link processes
validates the following law~---
a form of $\eta$-expansion~---
 (the law fails  for Milner's encoding into the $\pi$-calculus): 
 \[
   \abs y x y = x
   \enspace.
\]
In the call-by-value  $\lambda$-calculus this is a useful law, 
that holds because  substitutions  replace variables with values.


\section{Encoding in the Internal \pc}
\label{s:enc:pii}
We present the encoding in Section~\ref{ss:enc_pii}, together with an
optimised version of it, which is actually the main object of study to
establish soundness. We then proceed to establish 
validity 
of $\betav$-reduction and operational correspondence, both for the
optimised encoding, in
Section~\ref{s:oper:corr}. This allows us to derive soundness w.r.t.\
$\enfe$ of the
original (unoptimised) encoding in Section~\ref{s:sound}.
In Section~\ref{s:complete},  
we establish the completeness of the encoding, 
using the unique solution technique.

Some proofs are omitted from the main text, and 
are given in~\ref{a:soundness:pii}.

\subsection{The Encoding and its Optimised Version} 
\label{ss:enc_pii}

\begin{figure}[t]
  \begin{mathpar}
    \begin{array}{rcl}
\enca {\abs x M}&\deff& \bind p \outb p y . !\inp y {x,q}.\enc M q
\\[\mypt]
\enca x& \deff& \bind p \outb p y. \fwd y x
\\[\mypt]
      \enca {MN} &\deff& \bind p \new q \big(\enc M q
                         | \inp q y. \new r \big(\enc N r |
      \\[\myptSmall]
                \multicolumn{3}{r}{  ~~ \inp r w .\outb y {w',p'}.(\fwd {w'} w|\fwd {p'} p)\big)\big)}
    \end{array}
  \end{mathpar}
\caption{The encoding into \Intp} 
\label{f:enc_internal}
\end{figure}

Figure~\ref{f:enc_internal} presents the encoding into \Intp, derived
from Milner's encoding by removing the 
free outputs as explained in Section~\ref{s:enc:cbv}.
Process   $\fwd ab$  represents a \emph{link} (sometimes called forwarder; 
for readability we use the infix
notation $\fwd ab$ for the constant $\fwd{}{}$). It 
transforms all outputs at $a$ into outputs at $b$ (therefore $a,b$ are
names of the same sort). The body of $\fwd ab$ is replicated,
unless $a$ and $b$ are \emph{continuation names}.
The
definition of the constant $\fwd{}{}$ therefore is: 
$$
\begin{array}{rcl}
\fwd{}{} &\Defi &
\left\{ \begin{array}{l}
\bind{p,q} \inp p {x}.\outb {q} {y}
          . {\fwd{ y}{ x}}\\[\myptSmall]
\multicolumn{1}{r}{  
       ~\quad  \mbox{if $p,q$ are continuation names}
                   }              \\[\mypt]
\bind{x,y} ! \inp x {z,p}.\outb y {w,q}
          .({\fwd{ q}{ p}}|\fwd w z)  \\[\myptSmall]
\multicolumn{1}{r}{  
          ~\quad \mbox{if $x,y$ are value names} 
}\end{array} \right. 
\end{array}
 $$
(The distinction between continuation names and value names is not necessary, but simplifies the proofs.)

We recall some useful properties related to compositions of  links \cite{internalpi}. 

\begin{restatable}{lemma}{linkcomp}
    \label{l:fwd}
  We have: 

\begin{enumerate}
\item $\new q (\fwd p q|\fwd q r)\exn \fwd pr
$, for all continuation names $p,r$.
\item $\new y (\fwd xy|\fwd yz)\exn \fwd xz
  $, for all 
  value 
  names $x,z$.
\end{enumerate}
\end{restatable}




\paragraph{The Optimised Encoding}

In order to establish operational correspondence
(Proposition~\ref{p:opt_op}) we introduce an
  \emph{optimised encoding}: we remove
some of the internal transitions from the encoding of Figure~\ref{f:enc_internal}, as they prevent the
use of the expansion $\exn$. This allows us to guarantee that if
$M\reds N$,  the encoding of $N$ is faster than the encoding of $M$, in the sense that it
performs fewer internal steps before a visible transition. As a
consequence, the encoding of  a term in normal form is ready to perform a
visible transition.
The results about the optimised encoding are formulated using the
expansion preorder, which is useful for the soundness proof. For the
completeness proof, we use these results with $\bsim$ in place of
$\exn$.


We therefore relate $\lambda$-terms and
\Intp-terms via the optimised encoding $\qencba$, presented in
Figure~\ref{f:opt_encod}.
In the figure we assume that rules
 \nameDS{var-val} and  \nameDS{abs-val} have priority over the others; 
accordingly, 
in rules \nameDS{val-app},  \nameDS{app-val}, and   \nameDS{app},
terms    $M$ and $N$ should  not  be values.

The optimised  encoding is
obtained  from
that in Figure~\ref{f:enc_internal} by performing a few (deterministic) reductions, 
at the price of 
a more complex definition.  Precisely, in the encoding of application
we remove some of the initial
communications, including those with which a term signals to have become a value. 
To achieve this,  the encoding of an 
application
is split into several cases, depending on whether a subterm of the
application is a value or not.
This yields to a case distinction according to whether the components
in an application are values or not. This is close in spirit to the
idea of the \textit{colon translation} in~\cite{DBLP:journals/tcs/Plotkin75}.

The general idea of the optimised encoding can be illustrated on two
particular cases. For $\encba{\val M}$, the corresponding equation is
the result of unfolding the original encoding, and performing one
(deterministic) communication. 
In the case of $\encba{x\val}$, not only do we unfold the original
encoding and reduce along deterministic communications, but we also perform 
the administrative reductions that always precede the execution of $\enca{x\val}$.

\begin{figure*}[t]
\[
\begin{array}{rcll}
\equaDS{{var-val}}
{  \encba {x\val}} \defi{ \bind p \outb x {z,q}.(\encbv \val z|\fwd q p)}
  \\
\equaDS{{abs-val}}
{ \encba {(\abs x M)\val}} \defi{ \bind p \new {y,w} (\encbv {\abs x M} y | \encbv {\val} w
}
  \\  && \qquad
         | \outb y {w',r'}.(\fwd {w'} w|\fwd {r'} p))
  \\
  \equaDS{{val-app}}
{ \encba {\val M}} \defi{ \bind p \new y (\encbv \val y | \new r
  (\encb M r
  }
  \\  &&  \qquad
| \inp r w
         .\outb y {w',r'}.(\fwd {w'} w|\fwd {r'} p)))
  \\
\equaDS{{app-val}}
  { \encba {M\val} } \defi {\bind p \new q (\encb M q
  | \inp q y. \new w (\encbv \val w
  }\\  &&  \qquad
|  \outb y {w',r'}.(\fwd {w'} w|\fwd {r'} p)))
\\
\equaDS{{app}}
{ \encba {MN}} \defi {\bind p \new q (\encb M q 
  | \inp q y. \new r (\encb N r |
  }\\  &&  \qquad
\inp r w
          .\outb y {w',r'}.(\fwd {w'} w|\fwd {r'} p)))
  \\
\equaDS{{opt-val}}
{ \encba {\val}}\defi{ \bind p \outb p y . \encbv\val y}
\end{array} 
  \]
where $ \encbva\val$ is thus defined: \hfill $ $ 
\[
\begin{array}{rcll}
\equaDS{{opt-abs}}{ \encbva {\abs x M}}\defi{ \bind y  !\inp y {x,q}.\encb M q  \hskip 2cm}\\
\equaDS{{opt-var}}{ \encbva x }\defi {\bind y \fwd y x}
\end{array}
\]
Moreover,
in rules \nameDS{val-app} and  \nameDS{app-val}, $M$ is  not   a value; 
in  rule  \nameDS{app}   $M$ and $N$  are  not values.
 \caption{Optimised encoding into \Intp}
\label{f:opt_encod}
\end{figure*}



The next lemma builds on Lemma~\ref{l:fwd} to show that, on the processes obtained by the encoding into \Intp, 
links behave as substitutions.  We recall that $p,q$ are continuation
names, whereas $x,y$ are
value names. 
%

\begin{restatable}{lemma}{propenc}
\label{l:fwd:optim}
We have: 
\begin{enumerate}
\item $\new x (\encb M p| \fwd x y )\exn \encb {M\subst x y} p$.  \label{l:subst}
\item $\new p (\encb M p| \fwd p q )\exn \encb M q$.  \label{l:const}
\item $\new {y}(\encbv \val y |\fwd x y)\exn \encbv \val x$. \label{l:substvar}
\end{enumerate}  
\end{restatable}

The following lemma, relating the original and the optimised encoding,
allows us to use the latter to establish the soundness of the
former. 
The proofs of Lemmas~\ref{l:fwd:optim} and~\ref{l:opt_sound} are presented in~\ref{a:soundness}.
\begin{restatable}{lemma}{optsound}
  \label{l:opt_sound} 
 $\enca M\exn\encba M$, for all $M\in\Lao$.
\end{restatable}

\subsection{Operational Correspondence}\label{s:oper:corr}
Thanks to the optimised encoding, we can now formulate and prove the 
operational correspondence between the (optimised) encoding and the source \lterms. 

The proofs of the two following lemmas are presented in~\ref{a:soundness}.

We first establish that reduction in the $\lambda$-calculus yields
expansion for the encodings of the $\lambda$-terms.

\begin{restatable}[Validity of
  $\betav$-reduction]{lemma}{betaval}\label{l:beta:opt}
  For any
  $M,N$ in $ \Lao$, $M\longrightarrow  N$ implies that  for any $p$,
$\encb Mp\exn\encb Np$.
\end{restatable}




To prove operational correspondence, we need a technical lemma about the optimised encoding of terms
in eager normal form that are not values.

\begin{restatable}{lemma}{optstuck} \label{l:opt_stuck}
We have: 
$$
\begin{array}{rcl}
  \encb {\evctxt[x\val]} p\sim 
  \outb x {z,q}.(\encbv \val z|\inp q
y.\encb {\evctxt[y]}p).
\end{array}
 $$ 
\end{restatable}

 In the lemma below, recall that we identify processes or transitions
that only differ in the choice of the bound names. Therefore, when 
we say a process has exactly one immediate transition, we mean that there 
is a unique pair ($\mu$,$P$), up to alpha-conversion, 
  such that 
$\encb M p \arr\mu P$. 

\begin{restatable}[Operational correspondence]{proposition}{optop}\label{p:opt_op}
For any
  $M\in\Lao$ and fresh $p$, process  $\encb M p$ has exactly
  one immediate transition, and exactly one of the following clauses holds:
\begin{enumerate}
\item $\encb M p\arr{\outb py}P$  
and $M$ is a value, with $P=\encbv M y$; 
\item $\encb M p\arr{\outb x {z,q}} P$  and 
 $M=\evctxt[x\val]$ and 
$P\exn \encbv \val z|\inp q y.\encb {\evctxt[y]} p$;
\item $\encb M p\arr\tau P$ and there is $N$  with $M\red N$ and $P\exn \encb N p$.
\end{enumerate}
\end{restatable}

\begin{proof}
  We rely on Proposition~\ref{p:case:analysis} to distinguish two
  cases:
  \begin{itemize}
  \item  $M$ is in eager normal form.

    Then $(i)$ either $M$ is a value, and we are in the first case
    above, by definition; or $(ii)$ we are in the second case above,
    and we rely on Lemma~\ref{l:opt_stuck} to conclude.
  \item There exists $N$ such that $M\red N$,  We then use the
    property that is established in the proof of validity of
    $\betav$-reduction (Lemma~\ref{l:beta:opt}), namely that $\encb
    Mp\arr\tau\exn\encb Np$. 
  \end{itemize}
\end{proof}





The operational correspondence has 
two immediate consequences, regarding converging and diverging terms. 

\begin{lemma}\label{l:opt_enf} If $\encb M p\Arr{\mu}P$ and $\mu\neq
  \tau$, then $M$ admits an \enform $M'$ such that $\encb {M'} p
  \arr{\mu}P_0$ and $P \exn P_0$ for some $P_0$.
\end{lemma}
\begin{proof}
  By induction on the length of the reduction $\encb M
  p\Arr{\mu}P$.
  If $\encb M p \arr\mu P$ and $\mu\neq \tau$, by Proposition~\ref{p:opt_op}, $M$ is an \enform. Otherwise, there is $P'$ such
  that $\encb M p \arr\tau P'\Arr\mu P$; by Proposition~\ref{p:opt_op},
  there is $N$ such that $M\red N$ and $P'\exn \encb N p$.
  Therefore, since $P'\Arr\mu P$, there is $Q$ such that
  $\encb N p\Arr\mu Q$,  $P\exn Q$ and, by definition of expansion,
  the sequence of transitions from $\encb Np$ to $Q$ is shorter
  than the sequence from $P'$ to $P$.
  By the induction hypothesis, N admits an \enform $N'$, with
  $\encb{N'}p\arr\mu Q_0$ and $Q\exn Q_0$ for some $Q_0$.

  $N'$ being an \enform for $N$, it is also an \enform for
  $M$. Moreover, we deduce from $Q\exn Q_0$ and $P\exn Q$ that $P\exn Q_0$.
\end{proof}

\begin{lemma}\label{l:bot}
 $\encb M p  \bsim  \zero$ iff $M  \Uparrow$.
\end{lemma}
\begin{proof}
\begin{enumerate}
\item Suppose $\encba M\not\bsim\bind p \zero$. Then $M\arr\mu P$ for some 
$\mu\neq\tau$, and by Lemma~\ref{l:opt_enf}, $M$ has an eager normal
form (hence $M$ does not diverge).
\item Assume now $\encba M\bsim \bind p \zero$. By Proposition~\ref{p:opt_op}, $\encb Mp\arr\tau P$ for some $P$, and there is $N$ such that 
  $M\red N$, and $\encb N p\bsim P$. Since $\encba M\bsim\bind
  p\zero$, we have $P\bsim\zero$, thus $\encba N\bsim\bind p\zero$.
With this property, we can construct an infinite sequence of reductions
from $M$, thus $M\Uparrow$.
\end{enumerate}
\end{proof}

\subsection{Soundness}
\label{s:sound}

The structure of the 
proof of soundness of the encoding is similar to that for the
analogous property for 
Milner's call-by-name encoding with respect to Levy-Longo Trees \cite{cbn}. 
The details are however different, as in call-by-value both the encoding 
and the trees (the Eager Trees extended to handle $\eta$-expansion) are
more complex.


Using the operational correspondence, we then show that the observables 
for bisimilarity in the encoding $\pi$-terms  imply the observables for 
$\eta$-eager normal-form bisimilarity in the encoded $\lambda$-terms.
The delicate cases are those
in which a branch in the tree of the terms is produced~---
case \reff{ie:split} of Definition~\ref{enfbsim}~--- and where
 an $\eta$-expansion    occurs~---  thus
a variable is equivalent to an abstraction,  
 cases~\reff{lab:five} and~\reff{def:enfe:case:eta} of Definition~\ref{enfebsim}. 

For the branching,  we exploit a decomposition property on $\pi$-terms, roughly allowing
us to derive from the bisimilarity of two parallel compositions the componentwise 
bisimilarity of the single components. 
 For the $\eta$-expansion, 
if $\enca x \bsim\enca {\abs zM}$,
where   $M\converges \evctxt[x\val]$, 
 we use a coinductive argument to
derive $\val\enfe z$ and $\evctxt [y]\enfe y$, for $y$ fresh; 
from this we  then obtain
 $\abs zM \enfe x$.  
 
The following lemma
 allows us 
to decompose 
 an equivalence between two 
 parallel processes.
%
%
This result is used to handle equalities of the form  $\enca
{\evctxt[x\val]}\bsim \enca {\evctxt'[x\valp]}$, in order to deduce
equivalence between $\val$ and $\valp$ on the one hand, and between 
$\evctxt[y]$ and $\evctxt'[y]$ on the other.


\begin{lemma}\label{l:pref_decompose}
Suppose that $a$ does not occur free in  $Q$ or $Q'$, 
and one of the following holds:
\begin{enumerate}
\item $\inp a x .P|Q\bsim \inp a x . P'| Q'$;
\item $!\inp a x .P|Q\bsim !\inp a x . P'| Q'$.
\end{enumerate} 
Then we also have  $Q\bsim Q'$.
\end{lemma}

 
\begin{lemma}\label{l:bangy}
  Suppose $!y(z,q).\encb Mq \bsim !y(z,q).\encb Nq$, where $y$ does not
  occur in $M, N$.
  Then $\encb
  Mq\bsim\encb Nq$.
\end{lemma}
\begin{proof}
  We first observe that if $\encb Np\Arr{} P$, then there exists $N'$
  such that $N\red^* N'$, $P\exn\encb{N'}p$ and $\encb Np\bsim P$. This follows from
  operational correspondence (Proposition~\ref{p:opt_op}) and validity
  of $\betav$-reduction (Lemma~\ref{l:beta:opt}).
  
  Let us now prove the lemma. We play a transition $\arr{y(z,q)}$ on
  the left hand side. The answer on the right hand side leads to a
  process $P$ such that $\encb Nq\Arr{}P$ and
  $$!y(z,q).\encb Mq | \encb Mq \bsim
  !y(z,q).\encb Nq | P.$$
  By the observation above, we deduce that 
  $$!y(z,q).\encb Mq | \encb Mq \bsim
  !y(z,q).\encb Nq | \encb Nq.$$
  We can then conclude using Lemma~\ref{l:pref_decompose}.
\end{proof}

We now show that the only \lterms whose encoding is bisimilar to 
the encoding of some variable $x$ 
reduce either to $x$, or to a (possibly infinite) 
$\eta$-expansion of $x$.

\begin{lemma}\label{l:opt_eta}
  If $\val$ is a value and $x$ a variable,
  $\encbva \val\bsim \encbva x$ 
  implies that either $\val=x$ or $\val = \abs z {M}$, where the eager
  normal form of $M$ is of the form $\evctxt[x\valp]$, with
  $\encba {\valp}\bsim \encba z$ and $\encba {\evctxt[y]}\bsim\encba y$
  for any $y$ fresh.
\end{lemma}
\begin{proof}
We observe that if $y\neq x$, then
  $\encbva y \not\bsim \encbva x$; therefore if $\encbva \val \bsim \encbva x$ and
  $\val\neq x$, then $\val$ has to be an abstraction $\val=\abs z M$.
  In this case, by definition, we have:
\begin{itemize}
\item $\encb {\abs z M} p= \outb p y. !\inp y {z,q}.\encb M q $
   and
\item $\encb x p = \outb p y. !\inp y{z,q}.\outb x{z',q'}.(\fwd {z'}z|\fwd {q'}q)$.
\end{itemize}
Therefore $\encb M q \bsim \outb x{z',q'}.(\fwd {z'}z|\fwd {q'}q)$,
and by Lemma~\ref{l:opt_enf} and Proposition~\ref{p:opt_op}, $M$ has an eager normal form $\evctxt[x\valp]$. 
We have, using Lemma~\ref{l:opt_stuck} (since $\exn\protect{\subseteq}\bsim$):
$$
\begin{array}{rcl}
\encb M q &\bsim&  \encb {\evctxt[x\valp]} q \\[\myptSmall]
&\bsim& \outb x {z',q'}.(\encbv {\valp}
        {z'}|\inp {q'} y.\encb {\evctxt[y]}q)
        ,
\end{array}
$$
which gives
$$\encbv {\valp} {z'}|\inp {q'} y.\encb {\evctxt[y]}q\bsim 
\fwd {z'}z |\fwd {q'}q . $$
We observe that $z'$ does not occur free in $\inp {q'} y.\encb
{\evctxt[y]}q$ and $q'$ does not occur 
free in $\encbv {\valp}{z'}$. Furthermore, $\encbv {\valp}{z'}$ is prefixed by 
an input on $z'$. By applying Lemma~\ref{l:pref_decompose} twice, we
deduce 
$$\encbv {\valp} {z'}\bsim  \fwd {z'}z\text{\hspace{2ex} and
  \hspace{2ex}}\inp {q'} y .\encb {E[y]}q \bsim \fwd {q'}q
\enspace.$$ 

By definition, $\fwd {z'}z=\encbv z {z'}$, so we have
$\encb {\valp} {z'}\bsim \encb z {z'}$.

We also have
\begin{align*}
\inp {q'} y .\encb {\evctxt[y]}q &\bsim \fwd {q'}q =\inp {q'} y.\outb q {y'}.(\fwd {y'}y),
\end{align*}
which gives
\begin{align*}
\encb{\evctxt[y]}q &\bsim \outb q {y'}.(\fwd {y'}y) =\encb y q.
\end{align*}
\end{proof}

We can  now prove soundness of the encoding.

\begin{proposition}[Soundness]\label{l:sound}
For any  $M,N \in \Lao$, if $\enca M\bsim\enca N$ then $M\enfe N$.
\end{proposition}

\begin{proof}

Let $\R\defi\{(M,N)\st \encba M\bsim\encba N\}$; we show that $\R$ is an \enfbsim, 
and conclude by Lemma \ref{l:opt_sound}.
Assume $\encba M\bsim \encba N$. 
\begin{enumerate}
\item If $M\diverges$, 
by Lemma~\ref{l:bot}, for any fresh $p$, $\encb Mp\bsim\zero$. Thus
$\encb Np \bsim \zero $, and, by Lemma~\ref{l:bot} again, $N\diverges$.

\item Otherwise, $M$ and $N$ have \enforms $M'$ and $N'$; i.e.,  
$M\converges M'$ and $N\converges N'$. 
Therefore by Lemma~\ref{l:opt_sound} and 
validity of $\betav$-reduction, $\encba {M'}\bsim \encba{N'}$. 
Since $M'$ is in \enform, by Proposition~\ref{p:opt_op}, 
either $\encb {M'} p\arr{\outb x {z,q}}P$ or $\encb {M'} p\arr{\outb p y}P$, 
and likewise for $N'$. This yields two cases:
\begin{enumerate}
\item \label{p:sound:case:eval}$M'=\evctxt[x\val]$,
  $N'=\evctxt'[x\valp]$, and
  $$
    \encbv \val z|\inp q y.\encb {\evctxt[y]} p \quad
    \bsim\quad \encbv {\valp} z|\inp
  q y.\encb {\evctxt'[y]} p\enspace.
  $$
We observe that name $q$ does not appear free
   in $\encbv \val z$ or $\encbv {\valp} z$, hence by Lemma
   \ref{l:pref_decompose}, $\encbva\val \bsim\encbva{\valp}$, thus
   $\encba\val\bsim\encba\valp$.
   
Likewise, $z'$ does not appear free in either
$\inp q y.\encb {\evctxt[y]} p$ or $\inp q y.\encb {\evctxt'[y]} p$,
and both $\encbv \val z$ and $\encbv {\valp} z$ are prefixed by a 
replicated input on
$z$. Hence, by Lemma~\ref{l:pref_decompose}, 
$\encba {\evctxt[y]}\bsim\encba {\evctxt'[y]}$.
 
Therefore $\val\RR\valp$ and $\evctxt[y]\RR \evctxt'[y]$.

\item $M'$ and $N'$ are values. They can be  abstractions or variables. 
\begin{enumerate}
\item If both are abstractions $M'=\abs z M''$, $N'=\abs z N''$, and 
  $$!\inp y
  {z,q}.\encb {M''} q\bsim !\inp y {z,q}.\encb {N''}q,$$
  hence, by Lemma~\ref{l:bangy}, $\encba
  {M''}\bsim \encba{N''}$, which gives $M''\RR N''$.
\item If both are variables, as seen above, we necessarily have
  $M'=N'=x$.
We have $x\RR x$, thus we can conclude. 
\item Otherwise, assume $M'=\abs z M''$ and $N'=x$ without loss of generality.
Then $\encba {M'}\bsim\encba {N'}\bsim \encba x$.
By Lemma~\ref{l:opt_eta}, $M''\converges \evctxt[x\val]$ for some
$\evctxt$, $\val$, and also 
$\encba {\val}\bsim \encba z$ and
$\encba {\evctxt[y]}\bsim \encba y$ for some $y$ fresh.

Hence, $\val\RR z$,
$\evctxt[y]\RR y$ for some $y$ fresh, and we can conclude using case~\ref{def:enfe:case:eta} of Definition~\ref{enfebsim}.
\end{enumerate}
\end{enumerate}
\end{enumerate}
\end{proof}


\subsection{Completeness and Full Abstraction}
\label{s:complete}

 Suppose $M \enfe N$. Then there is an \enfbsim $\R$ such that $M\R N$.
 The completeness of the encoding can thus be stated as follows: given 
 $\R$ an \enfbsim, for all
 $(M,N)\in \R$, $\enca M \bsim\enca N$.

To increase readability of the proof, we first show  completeness for $\enf$, rather than $\enfe$.

\paragraph{Introducing systems of equations}
Suppose $\R$ is an eager normal-form bisimulation. 
We
define an (infinite) system of equations $\eqcbv$, solutions of which will be 
obtained from the encodings of the pairs in $\R$.
The definition of $\eqcbv$ is given on Figure~\ref{f:app:eq}.
We then use Theorem~\ref{thm:usol} and Lemma~\ref{t:transf:equations}
to show that $\eqcbv$ has a unique solution.

We assume an ordering on names and variables, so to be able to view
(finite) sets of these as tuples.  In the equations of
Figure~\ref{f:app:eq}, $\til y$ is assumed to be the ordering of
$\fv{M,N}$.  Moreover, if $F$ is an abstraction, say $\bind \tila P$,
then $\bind \tily F$ is an abbreviation for its uncurrying
$\bind{\tily,\tila}P$.

There is one equation $ X_{M,N} = E_{M,N}$ for each pair $(M, N)\in\R$.
The body $E_{M,N}$ is essentially the encoding
of the eager normal form
 (or absence
thereof)
 of $M$ and $N$,  with the
variables of the equations representing the coinductive hypothesis.
To formalise this, 
we extend the encoding of the $\lambda$-calculus to equation variables 
by setting
\[ 
 \enca {X_{M,N}}{} \deff \bind {p}  \app{X_{M,N}}{\tily,p} 
 \hskip .5cm \mbox{ ~~ where $\tily = \fv{M,N}$}
 \enspace.
\]

\paragraph{Systems \eqcbv and \eqcbvp}

We introduce two systems of equations,
\eqcbv
(Figure~\ref{f:app:eq}) and \eqcbvp (Figure~\ref{f:app:opt}).
On both figures, we provide the equations which are 
needed to handle $\enfe$.
The systems to handle $\enf$
are obtained by omitting some equations (precisely, the last two
equations in Figures~\ref{f:app:eq} and~\ref{f:app:opt}).

\begin{figure*}[t]
  \centering
\begin{align*}
&M\diverges\text{ and } N\diverges:
&X_{M,N} &= \bind {\til y} \enca \Omega \\
%
&M\converges \evctxt[x\val]\text{ and }N\converges\evctxt'[x\valp]:
&X_{M,N} &= \bind {\til y} 
\enca {(\abs z X_{\evctxt[z],\evctxt'[z]}
)
~(x~X_{\val,\valp}
)} \\
%
&M\converges x \text{ and }N\converges x:
&X_{M,N}&=\bind {\til y} \enca x \\
%
&M\converges \abs x M'\text{ and } N\converges \abs x N':
&X_{M,N}&=\bind {\til y}
\enca{\abs x X_{M',N'}
}\\
%
&M\converges x\text{, }N\converges \abs z N'\text{, }N'\converges\evctxt[x\val]:
&X_{M,N} &=\bind {\til y}
\enca{\abs z \left((\abs w X_{w,\evctxt[w]}) ~ (x~X_{z,\val}
)\right)}\\
%
&M\converges \abs z M'\text{, }M'\converges \evctxt[x\val]\text{, }N\converges x:
&X_{M,N} &=\bind {\til y}
\enca{\abs z \left((\abs w X_{\evctxt[w],w}) ~ (x~X_{\val,z}
)\right)}
\end{align*}
\caption{System \eqcbv of equations 
 (the last two equations are
  included only when considering $\enfe$)
}
\label{f:app:eq}
\end{figure*}

Given $(M,N)\in\R$, we now comment on the equation $ X_{M,N} = E_{M,N}$
in system \eqcbv.
%
The equation is 
 parametrised on  the free variables
of  $M$ and $N $ (to ensure that the body $E_{M,N}$ 
 is a name-closed
abstraction) and an additional continuation
name (as all encodings of terms).
Accordingly, we write $\til y$ for the ordering of $\fv{M,N}$.

\begin{enumerate}
\item 
If 
$M\diverges$ and
$N\diverges$, then the equation is
\begin{align*}
X_{M,N} = \bind {\til y} \enca \Omega 
\end{align*}
%
(We could use $\bind{\til y,p}\zero$ above, since the encoding of
a diverging term is bisimilar to $\zero$).

\item If $M\converges x$ and $N\converges x$, then
 the equation
 is 
 the encoding of $x$:
\begin{align*}
  X_{M,N}&=\bind {\til y} \enca x
                  = \bind {\til y,p}  \outb p z. \fwd z x
\end{align*}
Since $x$ is the \enform of $M$ and $N$, $x\in\til y$. 
Note that $\til y$ can contain more names, occurring free in $M$ or $N$.


\item If $M\converges \abs x M'$ and $N\converges \abs x N'$, then 
the equation encodes an abstraction whose body refers 
to the normal forms of  $M',N'$, via the variable $X_{M',N'}$:
\[
\begin{array}{rcl}
 X_{M,N}&=&\bind {\til y}
\enca{\abs x X_{M',N'}
}
\\[\myptSmall]
&= & \bind {\til y,p}
\outb p z .!\inp z
{x,q}.X_{M',N'}\param {\tilprime y,q},
\end{array}
\]
where $\tilprime y$ is the ordering of $\fv{M',N'}$.





\item\label{item:decomp:eqcbv} If $M\converges \evctxt [x\val]$ and
  $N\converges \evctxt' [x\valp]$, 
  we separate the evaluation contexts and the values, as in 
  Definition~\ref{enfbsim}. 
  In the body of the equation, this is achieved by: $(i)$ rewriting
  $\evctxt[x\val]$ into $(\abs z\evctxt [z])(x\val)$, for some fresh $z$,
  and similarly for $\evctxt'$ and $\valp$ (such a transformation is
  valid for $\enf$); and $(ii)$ referring to the variable for the
  evaluation contexts, $X_{\evctxt[z],\evctxt'[z]}$,
and to the variable for the values,  $X_{\val,\valp}$.
This yields the equation (for $z$ fresh):
\begin{align*}
 X_{M,N} = \bind {\til y} 
\enca {(\abs z X_{\evctxt[z],\evctxt'[z]}
)
~(x~X_{\val,\valp}
  )}
\end{align*}
\end{enumerate}

As an example, 
 suppose 
 $(I,\abs
 xM)\in\R$, 
where $I=\abs x x$ and $M= (\abs {zy}z) x x'$. 
We  obtain
the following equations: 
 (we have $\fv{M}=\{x,x'\}$, and we assume $x$ is before $x'$ in the
ordering of variables):
\begin{enumerate}
\item $\begin{aligned}[t]
X_{I,\abs x M}
&=\bind {x'}\enca {\abs x X_{x,M}
} 
\\
&=\bind {x',p} \outb p y .!\inp y {x,q}.X_{x,M}\param {x,x',q}
\end{aligned}$
\item $\begin{aligned}[t]
X_{x,M
}
&=\bind {x,x'}\enca x \\ 
&=\bind{x,x',p} \outb p y.\fwd y x
\end{aligned}$
\end{enumerate}

Before explaining how $\R$ yields solutions of the system,
we prove the following important law: 
\begin{restatable}{lemma}{solaux}\label{l:sol_aux}
If \evctxt\ is an evaluation context, $\val$ is a value,   $x$ is a
name and $z$ is 
fresh in \evctxt, then
$$\enca {\evctxt[x\val]} \bsim \enca {(\abs z \evctxt[z]) (x\val)}
.$$
\end{restatable}
\begin{proof}
  By induction on the evaluation context \evctxt. 
  When $\evctxt=[\cdot]$, we show that
  $\enca {x\val} \bsim \enca {I(x\val)}$ by  algebraic
reasoning. 
The other cases 
can be handled similarly. 
\end{proof}

\paragraph{Solutions of \eqcbv}
Having set the system of  equations for $\R$, we now  define 
solutions for it from the encoding of the pairs in $\R$.

We can view the relation $\R $  as an ordered
sequence of pairs (e.g., assuming some lexicographical  ordering). 
Then $\R_i$ indicates the tuple obtained by
projecting the   pairs in $\R$ onto the $i$-th component ($i=1,2$).
Moreover $(M_j,N_j)$  is the $j$-th pair in $\R$, and $\til{y_j}$ is 
$\fvars {M_j,N_j}$.

We write $\encOO\R$   for the closed abstractions 
resulting  from the
 encoding of $\R_1$, i.e., the tuple 
whose $j$-th component is 
$ 
\bind {\til{y_j}}\enca{M_j}$, and similarly for 
$\encTO\R$.
 \begin{definition}  
We define $
     \encleft \R\scdef \{\bind {\til y}\enca M\st \exists N, M\RR N$
  and $\til y =\fvars {M,N}\}$
. 
$\encright {\R}$ is defined similarly, based on the right-hand side of
the relation. 
\end{definition}

 We observe that if $\til y=\fv{M,N}$, then 
$\bind {\til y,p} \enc M p$ and
$\bind {\til y,p} \enc N p$
 are closed abstractions.

 As a direct consequence of Lemmas~\ref{l:opt_sound}
and~\ref{l:beta:opt}, we obtain the following result, which is used below.
\begin{lemma}[Validity of $\betav$-reduction for the original encoding]\label{l:beta}
    For any
  $M,N$ in $ \Lao$, $M\longrightarrow  N$ implies that  for any $p$,
$\enc Mp\bsim\enc Np$.
\end{lemma}

\begin{lemma}\label{l:sol}
 $\encleft{\R}$
 and $\encright{\R}$
 are solutions of the system of equations 
 \eqcbv.
\end{lemma}

\begin{proof}
We only show the property for $\R_1$, the  case for $\R_2$ is handled
similarly.

We show that each component of  $\encOO\R$ is solution of the
corresponding equation, i.e., 
for the $j$-th component 
we show 
$\bind { \til{y_j}} \enca{M_j}\bsim \Eqsing{M_j,N_j} {\encOO {\R}}$.

We reason by cases over the shape of the 
 eager normal form of $M_j,N_j$. 

\begin{itemize}
\item If $M\diverges$, we use Lemma~\ref{l:bot}, which gives us 
$\bind {\til y}\enca M\bsim\bind {\til y,p} \zero\bsim \bind{\til y}\enca \Omega$.
\item If $M\converges \evctxt[x \val]$, we have to show that: $$\enca M\bsim 
\enca {(\abs z \evctxt[z])(x\val)}\enspace.$$
By Lemma \ref{l:beta},
$\enca M\bsim \enca {\evctxt[x\val]}$;
we then
conclude by Lemma \ref{l:sol_aux}. 
\item If $M\converges \abs x M'$ (and $N$ also reduces to an
  abstraction), then: 
\begin{align*}
 \Eqsing {M,N}{\enca{\R_1}}\param {\til y,p}& = \outb p z.!\inp z {x,q}. \enca {M'} \param { q}\\[\myptSmall]
 &\bsim \enc M {\til y,p} \text{ (by Lemma \ref{l:beta})\enspace.}
 \end{align*}
\item If $M\converges x$ (and $N\converges x$), again:
\begin{align*}
 \Eqsing {M,N}{\enca{\R_1}}\param {\til y,p}& =
\enc x {p} \\[\myptSmall]
&\bsim \enc M {\til y,p}\text{ (by Lemma \ref{l:beta})\enspace.}
\end{align*}
\end{itemize}
\end{proof}

\begin{figure*}[t]
  \centering
\begin{align*}
&M\diverges\text{ and } N\diverges:
&X_{M,N} &= \bind {\til y,p} \zero \\
%
&M\converges \evctxt[xv]\text{ and }N\converges\evctxt'[xv']:
&X_{M,N} &= \bind {\til y,p} \outb x {z,q}.
             \\ &&&\qquad
(\XV_{\val,\valp}\param{z,{\tilprime y}}
|\inp q
  w.X_{\evctxt[w],\evctxt'[w]}\param{{\tilpprime y},p}) 
\\
%
&M\converges \val\text{ and }N\converges \valp:
&X_{M,N} &= \bind{\til y,p} \outb p y.\XV_{v,v'}\param{z,\tilprime y}\\
%
&\val= x \text{ and }\valp= x:
&\XV_{x,x}&=\bind{z,x}\fwd z x \\
%
&\val=  \abs x M\text{ and } \valp= \abs x N:
&\XV_{\abs x M,\abs xN}&=\bind {z,\til y}
!\inp z {x,q}.X_{M,N}\param{\tilprime y ,q}\\
%
&\val=x\text{, }\valp= \abs z N\text{, }N\converges\evctxt[x\val]:
&       
\XV_{x,\abs z N}
&=\bind {y_0,\til y}
                   !\inp {y_0}{z,q}. \outb x {z',q'}.
             \\ &&&\qquad
                   (\XV_{z,\val}\param{z',\tilprime y }
        |\inp {q'}{w}.X_{w,\evctxt[w]}\param{\tilpprime y,q})\\
%
&\val=\abs z M\text{, }M\converges\evctxt[x\val]\text{, }\valp= x:
&       
\XV_{\abs z M,x}
&=\bind {y_0,\til y}
                   !\inp {y_0}{z,q}. \outb x {z',q'}.
             \\ &&&\qquad
                   (\XV_{\val,z}\param{z',\tilprime y }
        |\inp {q'}{w}.X_{\evctxt[w],w}\param{\tilpprime y,q})
\end{align*}

\caption{System \eqcbvp of equations 
  (the last two equations are
included only when considering $\enfe$)}
\label{f:app:opt}
\end{figure*}
\paragraph{Unique solution for \eqcbv}

We rely on Theorem~\ref{thm:usol} to prove uniqueness 
of solutions for \eqcbv. 
The only delicate requirement is the one on divergence for the syntactic
solution.  
For this, we use  Lemma~\ref{t:transf:equations}. We thus introduce
 an auxiliary system of equations,  \eqcbvp, that extends \eqcbv, and
 whose syntactic solutions have no $\tau$-transition and hence trivially
 satisfy the requirement.  
 The definition of \eqcbvp{} is presented in Figure~\ref{f:app:opt}.
 Like the original system 
\eqcbv, so the new one 
 \eqcbvp\ is defined by inspection of the pairs in  $\R$;  
in 
 \eqcbvp, however, a pair of $\R$ may sometimes yield more than one
 equation. 
Thus, let $(M,N)\in \R$ with $\til y=\fvars {M,N}$ (we also
write $\tilprime y$ or $\tilpprime y$ for the free variables of the
terms indexing the corresponding equation variable).

\begin{enumerate}
\item When $M\Uparrow$ and  $N\Uparrow$, the equation is
\begin{align*}
  X_{M,N} = \bind{\til y,p} \zero
  \enspace.
\end{align*}
\item When $M\converges \val$ and $N\converges \valp$, we introduce a
  new  equation variable $\XV_{\val,\valp}$  and a new equation; 
this will allow us, in the following step (3), to 
perform some  optimisations. The equation is
%
\begin{align*}
 X_{M,N} = \bind {\til y,p} \outb p z .  \XV_{\val,\valp}\param{z,\tilprime y}
  \enspace,
\end{align*}
and we have, accordingly, the two following additional equations
corresponding to the cases where values are functions or variables:
%
%
\[
\begin{array}{rcl}
\XV_{\abs x M',\abs x N'}&=& \bind {z,\til y}!\inp z {x,q}
  . X_{M',N'}\param{{\tilprime y},q}
\\[\mypt]
\XV_{x,x}&=&\bind {z,x} \fwd z x
\end{array} 
\]

\item When $M\converges \evctxt[x\val]$ and $N\converges\evctxt[x\valp]$, we
refer 
to $\XV_{\val,\valp}$, instead of $X_{\val,\valp}$, so  to remove all
initial reductions in the corresponding equation for 
\eqcbv.  The first action thus becomes 
  an output: 
$$
  X_{M,N} = 
\bind {\til y,p} \outb x
              {z,q}.(\XV_{\val,\valp}\param{z,{\tilprime y}}
 |\inp q
       w.X_{\evctxt[w],\evctxt'[w]}\param{{\tilpprime y},p})
$$
\end{enumerate}



Lemmas~\ref{l:div_aux} and~\ref{l:div} 
are needed to apply Lemma~\ref{t:transf:equations}.
(In Lemma~\ref{l:div_aux},  `extend' is as by Definition~\ref{d:extend}.)

\begin{restatable}{lemma}{divaux}\label{l:div_aux}The system of equations \eqcbvp extends 
the system of equations \eqcbv.
\end{restatable}
\begin{proof}
  The new system \eqcbvp\ is obtained from \eqcbv\ by modifying the
  equations and adding new ones. More precisely, whenever $M$ and $N$
  are values, an additional equation is introduced, using a variable
  written $\XV$.
  A solution for the extended system yields a solution
  for the original system by looking only at equation variables which
  are not of the form $\XV_{V,V'}$.
\end{proof}

\begin{restatable}{lemma}{div}\label{l:div}
 \eqcbvp has a unique solution.
\end{restatable}
\begin{proof}
Divergence-freedom for the
 syntactic solutions
  of \eqcbvp\ 
holds because in
the equations each name (bound or free) can appear either only in inputs
or only in outputs. 
Indeed, 
  in the syntactic solutions 
  of \eqcbvp, linear names ($p,q,\dots$) are used exactly once in
  subject position, and non-linear names ($x,y,w,\dots$), when used in
  subject position, are either used exclusively in input or
  exclusively in output.

  As a consequence, since the labelled transition system is ground,
  no $\tau$-transition can ever be
performed, after any number of visible actions. 
Further,  \eqcbvp is  guarded. Hence   we can apply Theorem~\ref{thm:usol}.

\end{proof}

Hence, by Lemma~\ref{t:transf:equations}, \eqcbv\ has a unique solution.

%
We can observe that \eqcbv{} has equations of the form $X=\enca\Omega$
associated to a diverging \lterm. Such equations give rise to
\emph{innocuous divergences}, using the terminology of~\cite{usol}.  A
refined version of Theorem~\ref{thm:usol} is stated 
in~\cite{usol}, in order to handle such divergences; this would make
it possible to avoid using Lemma~\ref{t:transf:equations}, at the
cost of a more intricate setting.
Using Lemma~\ref{t:transf:equations} allows us to keep the framework
simpler.

A more direct proof of Lemma~\ref{l:div} would have been possible, by
reasoning coinductively over the \enfbsim defining
the system of equations. 

\begin{lemma}[Completeness for $\enf$]\label{l:complete_enf}
  $M\enf N$ implies $\enca M \bsim \enca N$,
for any $M,N \in\Lao$. 
\end{lemma}
 
\begin{proof}
Consider 
  an eager normal-form bisimulation $\R$, and 
the corresponding  systems of equations \eqcbv\ and  \eqcbvp.
  Lemmas~\ref{l:div} and~\ref{l:div_aux} allow us to apply
  Lemma~\ref{t:transf:equations} and deduce that \eqcbv\ has a unique
  solution.
By Lemma \ref{l:sol},
  $\encOO \R$ and $\encTO\R$ are solutions of \eqcbv. 
Thus, from $M\RR N$,
 we deduce
$\bind {\til y} \enca M  \bsim \bind{\til y} \enca N$, 
  where $\til y=\fvars {M,N}$. 
Hence also 
$ \enca M  \bsim \enca N$.



\end{proof}

\paragraph{Completeness for $\enfe$} 
The proof  for $\enf$ %
is extended to $\enfe$, maintaining  its
structure. We highlight the main additional reasoning steps. 




We enrich \eqcbv\ with the equations corresponding to the two
additional clauses of $\enfe$ (Definition~\ref{enfebsim}).
When $M\converges x$ and $N\converges \abs z N'$, where $N'\enfe xz$, 
we proceed as in case~\ref{item:decomp:eqcbv} of the definition of \eqcbv,
 given that  
$N\enfe \abs z \left( (\abs w\evctxt [w]) (x\val)\right)$. The
equation is:
\begin{align*}
X_{M,N} =\bind {\til y}
  \enca{\abs z \left((\abs w X_{w,\evctxt[w]}) ~ (x~X_{z,\val}
  )\right)}
  \enspace.
\end{align*}
%
We proceed likewise in the symmetric case.


In the optimised equations, 
we add the following equation (relating
values),  as well as its symmetric counterpart:
%
$$\XV_{x,\abs z N'}
=\bind {y_0,\til y}
    !\inp {y_0}{z,q}. \outb x {z',q'}. (\XV_{z,\val}\param{z',\tilprime y }
         |\inp {q'}{w}.X_{w,\evctxt[w]}\param{\tilpprime y,q})
         \enspace.$$

We follow the approach of Lemmas~\ref{l:div_aux} and~\ref{l:div} to
show  unique solution for \eqcbv.

The following lemma validates $\eta$-expansion for the encoding, and
is useful below. It is proved using algebraic reasoning and the
definition of links.
\begin{lemma}\label{l:eta_complete_aux}
 $\enca {\abs y {xy}} \bsim\enca x$. 
\end{lemma}

Finally, 
we 
prove that $\encOO \R$ and $\encTO \R$ are 
 solutions of \eqcbv. 
 Two additional cases are to be considered:
\begin{itemize}
\item If $M\converges x$ and $N\converges \abs y N'$, then:
\begin{align*}
\Eqsing {M,N}{\enca{\R_1}}\param {\til y,p} &=\outb p z. !\inp z {y,q}.\enc {(\abs w w)(xy)} q\\[\myptSmall]
&\bsim \outb pz .!\inp z {y,q}.\enc {xy}q\text{ (by Lemma \ref{l:sol_aux})}\\[\myptSmall]
&=\enc {\abs y xy}q\\[\myptSmall]
&\bsim\enc x p\text{ (by Lemma \ref{l:eta_complete_aux})}\\[\myptSmall]
&\bsim \enc {M}{\til y,p} \text{ (by Lemma \ref{l:beta})\enspace.}
\end{align*}
\item If $M\converges \abs y M'$, $M'\converges \evctxt[x \val]$ and $N\converges x$, then:
\begin{align*}
\Eqsing {M,N}{\enca{\R_1}}\param {\til y,p} &= \outb p z.!\inp z {y,q}.
\enc {(\abs w \evctxt[w])(x\val)} q\\[\myptSmall]
&\bsim \outb pz .!\inp z {y,q}.\enc {\evctxt[x\val]} q\text{ (by Lemma \ref{l:sol_aux})}\\[\myptSmall]
&= \enc {\abs y {\evctxt[x\val]}} q\\[\myptSmall]
&\bsim \enc M{\til y,p} \text{ (by Lemma \ref{l:beta})\enspace.}
\end{align*}
\end{itemize}

Given the previous results, we can reason as for the  proof of
Lemma~\ref{l:complete_enf} to establish
completeness.

\begin{proposition}[Completeness for $\enfe$] \label{l:complete}
  For any $M,N$ in $\Lao$, 
  $M\enfe N$ implies $\enca M \bsim \enca N$.
\end{proposition}

Combining Propositions~\ref{l:sound} and~\ref{l:complete}, and
Theorem~\ref{t:bisbc}, we deduce full abstraction for  
$\enfe$ with respect to barbed congruence. 

\begin{theorem}[Full Abstraction for $\enfe$]\label{thm:fa}
For any $M,N$ in $\Lao$, we have $M\enfe N$ iff $\enca M\wbc\IntpSmall \enca N$
\end{theorem}

\begin{remark}[Unique solutions versus up-to techniques]
\label{r:upto}
For Milner's encoding of call-by-name $\lambda$-calculus, 
the completeness part of the full abstraction result with respect to
Levy-Longo Trees~\cite{cbn} relies on
\emph{up-to techniques for bisimilarity}. 
Precisely, given a relation $\R$ on $\lambda$-terms that represents a
tree bisimulation,  one shows that the $\pi$-calculus encoding of 
$\R$  is a $\pi$-calculus bisimulation \emph{up-to context and
  expansion}. 

In the up-to technique, expansion is used
to manipulate the derivatives of two transitions so to bring up a
common context.   
Such up-to technique is not powerful enough for the 
 call-by-value encoding and the Eager
Trees
 because 
some of the required transformations would violate expansion (i.e., they
would require to replace a term by  a `less efficient' one).
An example of this is the law
proved in Lemma~\ref{l:sol},
that would have to be
applied from right to left so to 
implement the branching in clause
\reff{ie:split} of Definition~\ref{enfbsim}
(as a context with two holes). 

The use of the technique of  unique solution of equations allows us to
overcome the problem:  the law in Lemma~\ref{l:sol} 
and similar laws that
introduce 'inefficiencies' can be used (and they are indeed used, 
 in various places),  
as long as they 
 do not produce  new divergences.

\end{remark}



\section{Encoding into \alpi}
\label{s:localpi}
Full abstraction
with respect to  $\eta$-Eager-Tree equality also holds
 for Milner's simplest encoding, namely $\qencm$
 (Section~\ref{s:enc:cbv}),
provided that 
  the target
language of the encoding is taken to be  \alpi{} (see Section~\ref{s:subcalculi}). 
The adoption of  \alpi\ implicitly allows us to control capabilities,
avoiding violations of laws such as~\reff{eq:nonlaw} in 
the Introduction. 
In \alpi, bound output prefixes such as  $\outb a x .\inp x y$ are abbreviations for 
$\new x(\out a x| \inp x y )$.

\subsection{The Local \pc}
\label{alpiopsem}

We present here the results which make it possible for us to apply the
unique-solution technique to \alpi. 
%
The main idea is to exploit a characterisation of barbed
 congruence as ground bisimilarity. However, to obtain
 this, ground bisimilarity has to be set 
on top of a non-standard transition system, specialised to 
 \alpi~\cite{localpi}. The Labelled Transition System (LTS) is produced by the rules in Figure~\ref{f:lts:alpi};
 these modify ordinary transitions (the $\arr\mu$ relation)
by  adding \emph{static links}  $\alpilink a b$, which are
abbreviations  defined thus: 
\begin{mathpar}
  \alpilink a b \deff !\inp a{\til x}.\out b{\til x}
  \enspace.
\end{mathpar}
(We call them \emph{static} links, following the terminology in~\cite{localpi}, so to distinguish them from the 
links $\fwd ab$ used in \Intp, whose definition makes use of recursive
process definitions~--- static links only need replication.)

\begin{figure}
  \begin{mathpar}
    \inferrule{P\arr{\res\tild\out a\tilb}P'
      \and \til c\cap( \fnames P\cup\til d)=\emptyset
    }
    {P\alpiar{\res\tilc\out 
        a{\til c}}\res\tild(\alpilink \tilc \tilb | P')}
    \and
    \inferrule{P\arr{\inp a \tilb}P'}{P\alpiar{\inp a\tilb}P'}
    \and
    \inferrule{P\arr\tau P'}{P\alpiar\tau P'}
      \end{mathpar}
  \centering
\caption{The modified labelled transition system for \alpi}
\label{f:lts:alpi}
\end{figure}

Notations for the  ordinary LTS ($\arr\mu$) are transported onto the
new LTS ($\alpiar\mu$), yielding, e.g.,  transitions
$\alpiAr{\mu}$ and $\alpiAr{\hat\mu}$.

We write $\abt$  for (ground) bisimilarity on the new LTS, defined as 
 $\approx$ in Definition~\ref{d:bisimulation}, but using the new LTS in
place of the ordinary one.
Barbed congruence in \alpi, $\wbc\alpiSmall$, is defined as by
 Definition~\ref{d:bc} (on $\tau$-transitions,  which are the only
 transitions needed to define $\wbc\alpiSmall$, the new LTS and
 the original one coincide).

  We  present the definition of asynchronous (ground)
 bisimilarity, which is used in~\cite{localpi} to derive a
 characterisation of barbed congruence; asynchrony is needed because
 the calculus is asynchronous, and barbed congruence observes only
 output actions.

 \begin{definition}[Asynchronous bisimilarity]
  Asynchronous  bisimilarity, written $\abta$,   is the
  largest  symmetric relation $\R$ 
such that
  $P\R Q$ implies
  \begin{itemize}
  \item if $P\alpiar\mu P'$ and $\mu$ is not an input, 
    then there is $Q'$ s.t.\ $Q\alpiAr{\hat\mu}Q'$ and $P'\R Q'$,
    and
  \item if $P\alpiar{\inp a\tilb}P'$, 
    then either $Q\alpiAr{\inp a\tilb}Q'$ and $P'\protect{\R} Q'$ for some $Q'$, or
    $Q\Arr{}Q'$ and $P'\protect{\R} (Q' | \out a\tilb)$ for some $Q'$.
  \end{itemize}
\end{definition}

\begin{theorem}[\cite{localpi}]\label{t:ab:bc}
  On \alpi{} processes that are image-finite up to \bsim, relations 
 $\abta$ and $ \wbc\alpiSmall$ coincide.
\end{theorem}

To apply our technique of unique solutions of equations it
is however convenient to use  \emph{synchronous} bisimilarity. 
The following result  allows us to do so:

\begin{theorem}
\label{t:abt}
  On \alpi{} processes that are image-finite up to \bsim{} and have no
  input on a free name, 
relations 
 $\abt$ and $\abta$ 
coincide.
\end{theorem} 
\begin{proof}
By construction, $\abt\protect{\subseteq}\abta$.

To show that $\abta\protect{\subseteq} \abt$, 
we first establish a property about  output capability and
transitions. 
We say that $P$ \emph{respects output capability} if 
any free name used in input subject position in $P$
may not be used in output, either in subject or object position, in $P$.
We show that if $P$ respects output capability and $P\alpiar\mu P'$, 
then so does $P'$.

We reason on the type of the transition $P\arr{\mu'} P'$ from 
which $P\alpiar{\mu} P'$ is derived. For simplicity, 
we consider only monadic actions.

\begin{enumerate}
\item if $P\arr{\inp a c} P'$: then $c$ is fresh for $P$, and cannot
  be used in input in $P'$. The property hence holds.

\item if $P\arr\tau P'$: by hypothesis, the communication takes place at a restricted
  name (otherwise the name would have free occurrences in input and in
  output in $P$). If the transmitted name is restricted as well, there is
  nothing to prove. Otherwise, let
 us suppose, to illustrate the reasoning, that
%
$P= \new {a}(\inp a b.P_0|\outC a c )
\arr\tau\new a P_0\subst bc$ (the transmitted name is $c$). Because $c$
occurs free in output in $P$, it cannot occur in input, since $P$
respects output capability. Because we are in \alpi, the new
occurrences of $c$ created by the substitution $\subst bc$ are in
output position. So $c$ cannot occur in input position in $P'=\new a P_0\subst bc$. The proof
in the general case follows the same ideas.

\item if $P\arr{\outC a b}P'$, and $P\alpiar{\outb ac}\alpilink
  cb|P'$:
  because $P$ respects output capability, so does $P'$, as the
  transition $P\arr{\outC a b}P'$ does not introduce any new
  occurrence of names.  Since $b$ is used in output in $P$, it cannot
  be used in input in $P$, and hence $P'$ does not contain output
  occurrences of $b$. We deduce that $\alpilink cb|P'$ also
  respects output capability: indeed, $c$ is fresh and thus is used only
  in input, and $b$ is used only in output in $\alpilink cb$.
  
\item if $P\arr{\outb a b}P'$, and
  $P\alpiar{\bout ac}\new b(\alpilink cb|P')$. This case is simpler
  than the previous one, since $b$ is bound in the resulting process,
  and, as before, $c$ is fresh and is only used in input.
\end{enumerate} 

Thus, we can consider bisimulations containing only processes that respect 
output capability.

Now, assume $\R$ is an asynchronous bisimulation relation containing 
only processes that respect output capability. 
We show that $\R$ is also a synchronous bisimulation relation. 

The only
interesting case  is when  $(P,Q)\in\R$, 
$P\alpiar {\inp a \tilb} P'$ 
and 
$Q\alpiAr{~} Q'$ for some $Q'$ such that $P'\RR (Q'|\out a\tilb)$. 
We observe that $Q'|\out a \tilb$ can perform an output on $a$,  which
implies that $P'$ has a free occurrence of $a$  in output. Since by hypothesis
$a\notin\til b$, this means that $b$ occurs free in input and in
output in $P$, a contradiction.


Over processes that respect output capability,  
asynchronous bisimulation relations are synchronous bisimulation relations, thus  
$\abta $ and $\abt$ coincide.
This allows us to conclude, because
\alpi~processes that have no free input do respect output capability. 
\end{proof}

The property in Theorem~\ref{t:abt} is new~-- 
we are not aware of papers in the literature presenting it.
 It is a consequence of the fact that, under the hypothesis 
of the theorem, and with a ground transition system, 
the only 
 input actions in processes that can ever be produced are those
emanating from the
 links, and two tested processes, if bisimilar, must have  the
 same sets of (visible) links. 
We can moreover remark that in Theorem~\ref{t:abt}, the condition on inputs can be removed by
adopting an asynchronous variant of bisimilarity; however, the
synchronous version is easier to use in our proofs based on unique
solution of equations.

For any $M\in\Lao$ and $p$, process 
$\app{\encm M} p$ is  indeed image-finite up to \bsim and  has no free input.
From
Theorems~\ref{t:ab:bc} and~\ref{t:abt}, 
we therefore deduce that $\abt$ and $\wbc\alpiSmall$ coincide 
for processes obtained by the encoding $\qencm$. 

\subsection{Full abstraction in \alpi}\label{s:fa:alpi}
We now discuss full abstraction for Milner's encoding $\qencm$, when the target 
language is \alpi. 
The proof of is overall very similar to that of Theorem~\ref{thm:fa}. 

The systems of equations for \alpi{} are similar to the ones we
introduced for \Intp{} (Figures~\ref{f:enc_internal}
and~\ref{f:opt_encod}). They are presented in~\ref{a:alpi}.
The equations defining the first system 
are exactly
the same as in Figure~\ref{f:enc_internal}, only encoding $\qencm$
is used instead of $\qenca$.
The second system is an optimised version of the first. Again, it is defined
as in \Intp; the differences are that we use static links.



The modified LTS of Figure~\ref{f:lts:alpi} introduces additional
static links with respect to the ground LTS. When establishing the
counterpart of Lemmas~\ref{l:complete_enf} 
and~\ref{l:complete}, we need to
reason about
divergences, and must therefore show
that these links do not produce new reductions.
%


\begin{lemma}Let $P$ be an \alpi{} process such that $P$ has no divergences 
in the ground LTS. Then it has no divergences in the modified LTS for \alpi.
\end{lemma}
\begin{proof}
The replicated inputs guarding  static links 
created by an output transition in the modified LTS
are always at fresh names|the $\til c$ in Figure~\ref{f:lts:alpi}. 
Hence, no communication at the names in $\til c$ is
possible.
Furthermore, in the ground LTS, the additional inputs at the names in $\til c$ are with
fresh names, so they cannot generate new $\tau$ transitions.
\end{proof}


The encoding into \alpi{} is fully abstract.

\begin{theorem}
\label{t:faALpi}
$M\enfe N$ iff $  \encm M  \wbc\alpiSmall \encm N$, for any $M,N \in \Lao$. 
\end{theorem}


\section{Contextual equivalence and preorders}
\label{s:contextual}
We have presented full abstraction for $\eta$-Eager-Tree equality
taking
a `branching' behavioural equivalence, namely
 barbed congruence,  on the $\pi$-processes.
We show here the same result
for contextual
equivalence,  the most common 
 `linear' behavioural equivalence.
We also extend the results to   preorders.

We only discuss the encoding 
$\qenca$ into 
\Intp. Similar results however hold for 
the encoding $\qencm$ into
  \alpi.

\subsection{Contextual relations and traces}
\label{ss:preorders}

Contextual equivalence   is defined in the $\pi$-calculus 
analogously to its definition in the  
 $\lambda$-calculus (Definition~\ref{d:ctxeq}); 
thus, with
respect to barbed congruence, the bisimulation game on reduction is
dropped. 
Since we wish to handle  preorders, we also  introduce
the \emph{contextual preorder}. 

\begin{definition}
\label{d:ctx_pi}
Two \Intp\  processes 
$P,Q$ are in the \emph{contextual preorder}, written 
$P \ctxpre Q $, 
 if 
$C[P]\Dwa_{ a}$ implies $C[Q]\Dwa_{ a}$, 
for all contexts $C$. They are 
\emph{contextually equivalent}, 
written 
$P \ctxeqPI Q$, if both 
$P \ctxpre Q$  and $Q \ctxpre P$ hold.
\end{definition} 

As usual, these relations are extended to abstractions by requiring
instantiation of the parameters with fresh names. 
To manage  contextual  preorder and equivalence in proofs, we
exploit  
 characterisations of them  as trace inclusion and  equivalence.
%
%
 We define the traces of a process 
as follows:

\begin{definition} A (finite, weak, ground) trace is a 
finite sequence of visible actions $\mu_1,\dots,\mu_n$ such that 
for $i\neq j$ the bound names of $\mu_i$ and $\mu_j$ are all distinct 
and if $j<i$, the free names of $\mu_j$ and the bound names of $\mu_i $ 
are all distinct.
\end{definition}

For $s = \mu_1, \ldots, \mu_n$, 
 we write $P \Arr{s} $
if $P \Arr{\mu_1} P_1 \Arr{\mu_2}P_2 \ldots P_{n-1} \Arr{\mu_n}P_n$, for
some  processes $P_1, \ldots, P_n$.
In such a situation, we sometimes say that $s$ is a \emph{finite
  trace} of $P$.

\begin{definition}
\label{d:trace}
Two \Intp\ processes  $P,Q$ are in the \emph{trace inclusion}, written 
$P \trincl Q $, if $P \Arr{s} $ implies $Q \Arr{s} $, for each trace
$s$. They are 
\emph{trace equivalent}, 
written 
$P \treq Q $, if both 
$P \trincl Q $ and 
$Q \trincl P $ hold.
\end{definition} 

The following result is standard.

\begin{proposition}
\label{p:charac:traces}
In \Intp, relation  $ \ctxpre$ coincides with $ \trincl$, and 
relation $ \ctxeqPI$ coincides with $ \treq  $.
\end{proposition}

%




\subsection{A proof technique for preorders}
\label{ss:usol_pre}

We modify the technique of unique solution of equations to
reason  about preorders, precisely the 
 trace inclusion preorder. 

 We say that $\til F$ is  a \emph{pre-fixed
   point 
 for $\trincl$} of  
a system of equations   $\{\til X= \til E\}$ 
if $\til E[\til F]\trincl \til F$;
similarly,  $\til F$ is  a \emph{post-fixed
   point  for $\trincl$}
if 
$\til F\trincl \til E[\til F]$.
In the case of equivalence, 
the 
 technique of unique solutions   exploits symmetry arguments,
but 
symmetry  does not hold for preorders.  
We overcome the problem by referring to the
syntactic solution of the  system in an asymmetric manner. 
 This   yields the two lemmas below, intuitively stating  that the
 syntactic solution    of a  system
is its smallest pre-fixed point, as well as, under the  divergence-freeness
hypothesis, its greatest
post-fixed point. 

 In~\cite{usol}, in order to prove that a system of equations has a unique solution, 
we need to extend the transitions to \emph{equation expressions} (i.e., contexts). 
For the same reasons, here we consider that contexts or equations may perform transitions, 
obtained as for the LTS, and assuming the hole does not perform any action.  
This is extended to traces. 
Hence, 
if $E\param{\til a}$ has the trace $s$ (written $E\param{\til a}\Arr
s$), then for any  $\til F$,  
$E[\til F]\param{\til a}\Arr s$.  
For more details we refer the reader to \cite{usol}.

\begin{restatable}[Pre-fixed points,  $\trincl$]{lemma}{minsol}\label{least}
Let $\Eeq$ be a  
system of equations, 
and $\KEE$ its syntactic solution. 
If  $\til F$ is 
a pre-fixed point  for $\trincl$ of $\Eeq$, 
then $\KEE \trincl \til F$.
\end{restatable}
\begin{proof}
  Consider a finite trace $s$ of $\KEi i\param{\til a}$. As it is finite, there
  must be an $n$ such that $s$ is a  trace of $E_i^n\param{\til a}$, hence it is also
  a trace of $E_i^n[\til F]\param{\til a}$. 
  From $\til E[\til F]\trincl \til F$, 
  by congruence if follows that $E_i^{n+1}[\til F]\trincl E_i^n[F_i]$, hence 
  also $E_i^{n+1}[\til F]\trincl F_i$. 
%
Hence, $s$ is a
  trace of $F_i\param{\til a}$, and we can conclude by
  Proposition~\ref{p:charac:traces}. 
\end{proof}

\begin{restatable}[Post-fixed points,  $\trincl$]{lemma}{usoltrace}\label{greatest}
Let $\Eeq$ be a guarded 
system of equations, 
and $\KEE$ its syntactic solution. 
Suppose $\KEE$ has no divergences. 
If  $\til F$ is 
a post-fixed point  for $\trincl$ of $\Eeq$, 
then
$\til F \trincl \KEE$.
\end{restatable}

The proof of Lemma \ref{greatest}
is  similar to the 
proof of Theorem~\ref{thm:usol} (for bisimilarity). 
Details of the proof are given in~\ref{a:trace}. 

The following proof technique makes it possible to avoid referring to
the syntactic solution of a system of equations, which is sometimes
inconvenient.

\begin{theorem}
\label{thm:usol_pre}
Suppose that $\Eeq$ is a guarded 
system of equations 
 with a 
divergence-free 
 syntactic
solution.
If $\til F$ (resp.\ $\til G$) is a pre-fixed point (resp.\ post-fixed
point)   for $\trincl$ of $\Eeq$, 
then
 ${\til G}\trincl {\til F}$. 
\end{theorem}

\begin{proof}
Apply Lemma~\ref{greatest} to $\til F$ and Lemma~\ref{least} to $\til
G$: this gives $\til G\trincl \KEE \trincl \til F$.
\end{proof}


We can also extend
Lemma~\ref{t:transf:equations} to preorders.
Given two systems of equations $\Eeq$ and $\Eeq'$, we say that
 $\Eeq'$ 
\emph{extends $\Eeq$ with respect to a given preorder} 
 if there exists a fixed set of indices $J$ such that:
\begin{enumerate}
\item
  any pre-fixed point 
 of
$\Eeq$ for the preorder can be obtained from a pre-fixed point 
of $\Eeq'$ (for the same preorder) by removing the
components corresponding to indices in $J$;

\item 
the same as (1) with post-fixed points in place of pre-fixed points.
\end{enumerate}

\begin{lemma}\label{t:transf:preorders}
Consider two systems of equations $\Eeq'$ and $\Eeq$ 
where $\Eeq'$
extends $\Eeq$ with respect to $\trincl$. 
Furthermore, suppose $\Eeq'$ is  guarded and has a divergence-free 
syntactic solution. 
If $\til F$ is a pre-fixed point  for $\trincl$ of $\Eeq$, and 
 $\til G$  a post-fixed point  for $\trincl$ of $\Eeq$, 
then
 ${\til G}\trincl {\til F}$. 
\end{lemma}

\paragraph{Unique solution for trace equivalence}
 Theorem~\ref{thm:usol_pre} gives the following property:

\begin{corollary}\label{a:t:usol_trace}
In $\Intp$, a weakly guarded system of equations  whose 
syntactic solution does not diverge 
has a unique solution for $\treq$.
\end{corollary}

If $\til F \treq \til E[\til F]$ 
and $\til G \treq \til E [\til G]$, this gives, by applying Theorem~\ref{thm:usol_pre} 
twice,  $\til F\trincl \til G$ 
and $\til G\trincl \til F$,  hence $\til F\treq \til G$.

%

%


\subsection{Full Abstraction}
\label{ss:fa_preorders}

The preorder on $\lambda$-terms induced by the contextual preorder
 is  
 \emph{$\eta$-eager normal-form similarity}, $\esim$. It is   obtained by
imposing that $M \esim N$ for all $N$, whenever $M$ is  divergent.
Thus, with respect to the bisimilarity relation $\enfe$, we only have to change 
clause (1) of  Definition~\ref{enfbsim}, 
by  requiring 
 only $M$ 
 to be divergent.

\begin{definition}[$\eta$-eager normal-form similarity] A relation $\R$ between \lterms is an \enfse if, whenever $M\R N$, one of the following holds:
\begin{enumerate}
\item $M$ diverges
\item $M\reds \evctxt[x\val]$ and $N\reds \evctxt'[x\valp]$ for some $x$, $\val$, $\valp$, $\evctxt$ and $\evctxt'$ such that $\val\R \valp$ and $\evctxt[z]\R \evctxt'[z]$ for some $z$ fresh in $ {\evctxt,\evctxt'}$
\item $M\reds \abs x M'$ and $N\reds \abs x N'$ for some $x$, $M'$, $N'$ such that $M'\R N'$
\item $M\reds x$ and $N\reds x$ for some $x$
\item $M\reds x$ and $N\reds \abs z \evctxt[x\val]$ for some $x$, $z$, $\val$ and $\evctxt$ such that $z\R \val$ and $y\R \evctxt[y]$ for some $y$ fresh in $ \evctxt$
\item $N\reds x$ and $M\reds \abs z \evctxt[x\val]$ for some $x$, $z$, $\val$ and $\evctxt$ such that $\val\R z$ and $\evctxt[y]\R y$ for some $y $ fresh in $ \evctxt$
\end{enumerate}
$\eta$-eager normal form similarity is the largest \enfse.
\end{definition}



\begin{theorem}[Full abstraction on preorders]
\label{t:preorders}
For any $M,N \in \Lao$, we have
$M\esim N$ iff $ \enca M \ctxpre \enca N$.
 \end{theorem} 

 The structure of the proofs is similar to
 that for bisimilarity, 
 using however Theorem \ref{thm:usol_pre}. 
We discuss the main aspects of  the soundness and the completeness. 

\paragraph{Soundness for trace inclusion}
We show that 
$\enca M\trincl \enca N$ implies $M\esim N$.
The proof follows the same lines as the proof from Section \ref{s:sound}: we define the relation $\R\scdef \{(M,N)\st \encba M\trincl \encba N\}$, and show that it is an \enfse. The proof carries over similarly, using the equivalents for 
trace inclusion of Lemmas~\ref{l:opt_sound}, \ref{l:opt_enf}, \ref{l:bot} and \ref{l:opt_eta} and
Proposition~\ref{p:opt_op}. 

\paragraph{Completeness for trace inclusion}
Given an \enfse $\R$, we define a system of equations \eqcbv as in
Section \ref{s:complete}. The only notable difference in the definition 
of the equations is in the case
where $M\R N$, $M$ diverges and $N$ has an \enform. In this case, we
use the following equation instead:
\begin{equation}\label{e:tr}
X_{M,N}=\bind {\til y} \enca \Omega
\enspace.
\end{equation}
As in Section \ref{s:complete}, we define a system of
guarded equations \eqcbvp\ 
 whose syntactic solutions do not diverge. 
Equation~\reff{e:tr} is replaced with $X_{M,N}=\bind {\til y,p} \zero$.

Exploiting Lemma~\ref{t:transf:preorders}, 
we can use unique solution for preorders
 (Theorem~\ref{thm:usol_pre}) 
 {with \eqcbv instead of \eqcbvp.} 

Defining $\encleft \R$ and $\encright \R$ as previously, we need to
prove that $\encleft\R \trincl\Eq \R{\encleft\R}$ and 
$\Eq \R {\encright\R}\trincl \encright\R$. The former result is
established along the lines of the analogous result in
Section~\ref{s:complete}: indeed, $\encleft\R$ is a solution of
\eqcbv for $\bsim$, and $\treq$ is coarser than $\bsim$. 

For the latter, the only difference is due to equation \reff{e:tr},
when $M\R N$, and $M$ diverges but not $N$.
In that case, we have to prove that $\enca \Omega\trincl \enca N$, which
follows easily because 
the only trace of $\enca\Omega$ is the empty one, hence
$\enc \Omega p \trincl P$ for any $P$.

We can then derive full abstraction for contextual equivalence as a corollary.

\begin{corollary}[Full abstraction for $\ctxeqPI$] 
\label{t:faCTXpiI}
  For any $M,N$ in $\Lao$, 
  $M\enfe N$ iff $\enca M \ctxeqPI \enca N$.
\end{corollary}





\section{Conclusion}
\label{s:concl}

In the paper we have studied the main question raised in Milner's
landmark paper on functions as $\pi$-calculus processes, which is
about the equivalence induced on $\lambda$-terms by their process 
encoding.  We have focused on call-by-value, where the problem was
still open; as behavioural equivalence on $\pi$-calculus we have 
taken contextual equivalence and  barbed congruence (the most common
`linear' and 'branching' equivalences). 

First we have shown that some expected equalities 
for open  terms fail under Milner's encoding. We have considered two
ways for overcoming  this issue: rectifying the encodings (precisely,
avoiding free outputs); restricting the target language to \alpi, 
so to
control the capabilities of exported names.
We have proved that, in both cases, the equivalence induced  is 
Eager-Tree equality, modulo $\eta$
(i.e., Lassen's \enfbsim). 

We have then introduced a preorder on these trees, and 
 derived similar full abstraction results  for them with respect to 
the contextual preorder on $\pi$-terms. 
The paper is also a test case for 
the technique of
unique solution of equations (and inequations),
which is essential in all our
 completeness proofs.

Lassen had introduced  Eager Trees as the call-by-value analogous of
Levy-Longo and B{\"o}hm Trees. 
The results in the paper confirm the claim,  on process encodings of
$\lambda$-terms: it was known that for (weak and strong) call-by-name, the  equalities
induced are  those of  Levy-Longo Trees and   B{\"o}hm Trees~\cite{xian}.

For controlling capabilities, we have used \alpi. 
Another possibility would have been to use a  type system. 
In this case however, the technique of unique solution of equations
needs to be extended to typed calculi. We leave this for future work.

We also leave for future work a thorough comparison between the technique
of unique solution of equations and  techniques  based on enhancements
of the bisimulation proof method (the ``up-to'' proof techniques),
including if and how our completeness results can be
derived using the latter techniques. (We recall that  the ``up-to''
proof techniques are used in the completeness proofs  with
respect to Levy-Longo Trees and   B{\"o}hm Trees for the
\emph{call-by-name} encodings. 
We have discussed the problems with
call-by-value in Remark~\ref{r:upto}.)  In any case, even if 
other solutions existed, for this specific problem the unique solution 
technique appears to provide an elegant and natural framework to 
carry out the proofs.

For our encodings  we have used the polyadic $\pi$-calculus; Milner's
original paper \cite{milner:inria-00075405} used the monadic calculus (the polyadic $\pi$-calculus
makes the encoding easier to read; it had not been introduced at the
time of \cite{milner:inria-00075405}). We believe that polyadicity
does not affect the results in the paper (the possibility of
autoconcurrency breaks full abstraction of the 
 encoding of the polyadic
$\pi$-calculus into the monadic one, but autoconcurrency does not
appear in the encoding of $\lambda$-terms).

In the call-by-value strategy we have followed, the function is
reduced before the argument in an application. Our results can be
adapted to the case in which the argument runs first, changing the
definition of evaluation contexts. The parallel call-by-value, in
which function and argument can run in parallel (considered in
\cite{encodingsmilner}), appears more delicate, as we cannot rely on the usual
notion of evaluation context.


 Interpretations of $\lambda$-calculi into 
 $\pi$-calculi 
appear  related to
game
semantics~\cite{BHYseqpi,DBLP:conf/fpca/HylandO95,DBLP:journals/tcs/HondaY99}.
In particular, for untyped call-by-name they both allow us to derive 
 B{\"o}hm Trees and  Levy-Longo 
Trees~\cite{DBLP:journals/tcs/KerNO03,DBLP:journals/tcs/OngG04}.
In this respect, it would be interesting to see whether the
relationship between $\pi$-calculus and Eager Trees studied  in 
this paper could help to establish similar relationships in game
semantics. 



\bibliographystyle{plainurl}
\bibliography{efap.bib}

\clearpage
\appendix
\section{List of symbols for behavioural relations}
\label{a:tab}
The following table summarises the notation used for the 
equivalences and preorders  in the paper.

\[
\begin{array}{lll}
\enf&\text{eager normal-form bisimilarity}&\text{ (Definition \ref{enfbsim}
)}\\
\enfe&\text{$\eta$-eager normal-form bisimilarity}&\text{ (Definition \ref{enfebsim}
)}\\
\esim&\text{$\eta$-eager normal-form similarity}&\text{ (Section \ref{ss:fa_preorders}
)}\\
\wbc\LL&\text{barbed congruence in   $\LL$}&\text{ (Definition \ref{d:bc}
)}\\
\bsim &\text{(weak) bisimilarity}&\text{ (Definition \ref{d:bisimulation}
)}\\
\lexn& \text{expansion}&\text{(Definition~\ref{d:expan})}
\\
\ctxeq &\text{contextual equivalence in   $\Lambda$}&\text{ (Definition~\ref{d:ctxeq})}\\
\ctxeqLL\LL &\text{contextual equivalence in   $\LL$}&\text{ (Section~\ref{ss:preorders})}\\
\ctxpre &\text{contextual preorder in   $\Intp$}&\text{ (Section~\ref{ss:preorders})}\\
\treq&\text{trace equivalence}&\text{ (Section 
\ref{ss:preorders})}\\
\trincl&\text{trace inclusion}&\text{ (Section 
\ref{ss:preorders})}
\end{array}
 \]
where $\LL$ is supposed to be a subcalculus of $\pi$; in the paper we have considered
 \Intp\ and \alpi.




\section{Properties of Milner's encoding (Section~\ref{s:enc:cbv})}
\label{a:milner}

\nonlaw*

\begin{proof}
For simplicity, we give the proof 
when $\val=y$, and for encoding $\qencm$. 
The same can be shown for an arbitrary value $\val$ and for 
the encoding $\qencm'$, through similar calculations. 
We use algebraic laws of the equivalence $\wbc\pi$, or of its associated 
proof techniques, to carry out the calculations (cf.~\cite{SW01a}).

\begin{align*}
  \encma {xy} p & = \new q (\out q x| \inp q \cv.\new r(\out  r y| \inp r w.\out \cv {{w,p}}))\\
    &\wbc\pi \new r(\out r y|\inp r w.\out x {{w,p}})\\
  & \wbc\pi\out x {y,p}\\
  \encma {I(xy)} p&\wbc\pi  \new q (\encma I q | \inp q \cv. \new r(\encma {xy} r|\inp r w.\out \cv {w,p})) \\
  &\wbc\pi \new q (\outb q \cv.!\inp u {z,q'}.\out {q'} z| \inp q \cv.\new r(\encma{xy}r|\inp r w.\out\cv{w,p}))\\
      &\wbc\pi \new \cv(!\inp \cv {z,q}.\out q z| \new r (\out x {y,r}|\inp r w.\out \cv {w,p}))\\
  &\wbc\pi \new r(\out x {y,r} |\inp r w.\new u (!\inp u{z,q}.\out q z|\out u {w,p} ))\\
    &\wbc\pi \new r (\out x {y,r}|\inp r w.\out p w)\\
  &=\new q( \out x {y,q} | \inp q z.\out p z  ) \\
&=\new q (\out x {y,q}. \fwd q p)
\end{align*}
\end{proof}



\newcommand{\lhs}{\ensuremath{\mathsf{lhs}}}
\newcommand{\rhs}{\ensuremath{\mathsf{rhs}}}
\newcommand{\wires}{\ensuremath{\mathsf{W}}}

\section{Soundness proof (Section~\ref{s:enc:pii})}
\label{a:soundness:pii}

\subsection{Properties of the optimised encoding}\label{a:calc:encpiI}

\linkcomp*

\begin{proof}
We first show the two following laws:
\begin{align*}
\new q (\fwd p q|\fwd q r)&\sim \inp p x. \new q(\outb q y.\fwd y x | \fwd q r)\\
&\exn \inp p x. \new {y}(\fwd y x| \outb r z. \fwd z y)\\
&\sim \inp p x. \outb r z. \new y (\fwd z y| \fwd y x)
\end{align*}
and 
\begin{align*}
\new y (\fwd xy|\fwd yz)&\sim !\inp x {p,x'}. \new y(\outb y {q,y'}.(\fwd {y'} {x'} |\fwd q p) | \inp y {q,y'} .\outb z {r,z'}.(\fwd {z'}{y'}|\fwd r q))\\
&\exn !\inp x {p,x'}. \new {q,y'}.(\fwd {y'} {x'}| \fwd qp | \outb z {r,z'}.(\fwd {z'}{y'}|\fwd r q))\\
&\sim !\inp x {p,x'}.\outb z {r,z'}. (\new {q} (\fwd rq| \fwd qp)|\new {y'}(\fwd{z'}{y'}|\fwd {y'}{x'}))
\end{align*}

We define a relation $\R$ as the relation that contains, for all continuation names $p, r$ and 
for all value 
names $x,z$, 
 the following pairs:
\begin{enumerate}
\item $\new q(\fwd p q |\fwd q r)$ and $\fwd p r$\\
$\new y(\fwd xy|\fwd yz)$ and $\fwd x z$
\item $\outb r z. \new y (\fwd z y| \fwd y x)$ and $\outb r z . \fwd z x$\\
$\outb z {r,z'}. (\new {q} (\fwd rq| \fwd qp)|\new {y'}(\fwd{z'}{y'}|\fwd {y'}{x'}))$ and 
$\outb z {r,z'}.(\fwd rp|\fwd {z'}{x'})$\\ (for all non-continuation names $x$ or $x',z'$)
\end{enumerate}

We show that this is an expansion up to expansion and contexts, using the previous laws (each of those processes has only one possible action). 
\end{proof}

In the proofs below, we shall use the following properties about the
optimised encoding, without referring explicitly to this lemma. 
\begin{lemma}~
  \begin{itemize}
  \item For any $M, p$, $\encb Mp$ cannot perform any input at $p$.
  \item For any $\val, y$, the only transition that $\encbv{\val}y$
    can do is an input at $y$.
  \item For any $M, x, p$, if $x\in\fvars M$, then $x$ appears in
    $\encb Mp$ only in output subject position.
  \end{itemize}
\end{lemma}
The
properties above follow by an easy induction.
The last one allows us to use the distributivity properties
of private replications~\cite{SW01a} when reasoning algebraically about
the encoding of $\lambda$-terms.

\propenc*

\begin{proof}
\newcommand{\IHone}{(IH1)}
\newcommand{\IHtwo}{(IH2)}
\newcommand{\IHthree}{(IH3)}
\newcommand{\Ltwo}{(L2)}
\newcommand{\Lone}{(L1)}
\newcommand{\Lthree}{(L3)}

  We establish the conjunction of the three following properties:
  \begin{enumerate}
  \item[\Lone]\label{conj1} $\new x (\encb M p| \fwd x y )\exn \encb {M\subst x y} p$,
    and, if $M$ is equal to some value $\val$,
    we have $\new x(\encbv{\val}{y_1} | \fwd xy)\exn
    \encbv{\val\subst xy}{y_1}$, for any $y_1$. 
\item[\Ltwo]\label{conj2} $\new p (\encb M p| \fwd p q )\exn \encb M q$.  
\item[\Lthree]\label{conj3} If $M$ is equal to some value $\val$, $\new {y}(\encbv \val y |\fwd x y)\exn \encbv \val x$. 
  \end{enumerate}
Indeed, as we show below, there are dependencies between these three properties, which
prevent us from treating them separately.

We reason by induction over $M$, 
and introduce some notations.
In each case, we use \IHone{} to refer
to property~\Lone{} of the induction hypothesis, and similarly for
\IHtwo{} and \IHthree.

To reason about \Lone, we  set $\lhs = \new x(\encb Mp | \fwd xy)$ and $\rhs =
\encb{M\subst xy}p$.

\bigskip
\noindent
\textbf{First case: $M=z$.}

\textbf{\Lone.}
Then $\lhs =\new x(\encb Mp | \fwd xy)
      = \new x(\bout p{y_1}.\fwd{y_1}z | \fwd xy)$. There are two sub-cases.
      \begin{itemize}
      \item $M=z\neq x$. Then $\lhs 
        \sim \bout p{y_1}.\fwd{y_1}z=\rhs$ because $\new x\fwd
        xy\sim\zero$.

              Since $M$ is a value, we must also check that $\new x(\encbv
        z{y_1} | \fwd xy)\exn \encbv{z\subst
          xy}{y_1}$. Lemma~\ref{l:fwd} allows us to show this.

      \item $M=x$. Then
        \begin{eqnarray}
        \lhs& \sim& \bout p{y_1}.\new x(\fwd{y_1}x |
                    \fwd xy) \label{step:1}\\
          &\exn& \bout p{y_1}.\fwd{y_1}y = \rhs\label{step:2}
        \end{eqnarray}
        \rref{step:1} holds because the only transition \lhs{} can do
        is the output at $p$. \rref{step:2} follows from
        Lemma~\ref{l:fwd}.

        Again, since $x$ is a value, we have to show the corresponding
        property, which amounts to show $\new x(\fwd{y_1}x | \fwd
        xy)\exn \fwd{y_1}y$, which is given by Lemma~\ref{l:fwd}.
      \end{itemize}

      \textbf{\Ltwo.}
      We write in general, for any value $\val$ (the following reasoning is 
      also used below, in the case where $M$ is an abstraction):
          \begin{eqnarray}
      \new p(\encb{\val}p | \fwd pq)
      & = &
            \new p(\bout py.\encbv\val y |\fwd pq) \nonumber\\
      &\exn & \new y(\encbv\val y | \bout q{y'}.\fwd{y'}y)\label{step12}\\
      &\sim & \bout q{y'}.\new y(\encbv\val y|\fwd {y'}y)\label{step13}
    \end{eqnarray}
    Step~\rref{step12} holds because the first process
    deterministically reduces to the second, and step~\rref{step13}
    holds because the only action the process can do is the output at
$q$.

Now since $\val=z$, we have  $\encbv zy = \fwd
yz$, and we obtain $\bout q{y'}.\new y(\fwd yz|\fwd{y'}y)$, a process
that expands $\enc zq$ by Lemma~\ref{l:fwd}.

\textbf{\Lthree.} We check that we do have $\new y(\fwd yz | \fwd xy) \exn\fwd xz$
by Lemma~\ref{l:fwd}.

\bigskip

\noindent
\textbf{Second case: $M=\lambda z.M'$.}

\textbf{\Lone.} We distinguish two cases.
      \begin{itemize}
      \item If $z\neq x$, then we can write
        \begin{eqnarray}
        \lhs &= &\new x(\bout p{y_1}.!y_1(z,q).\encb{M'}q | \fwd
     xy)\nonumber\\
             &\sim& \bout p{y_1}.!y_1(z,q).\new x(\encb{M'}q | \fwd xy)
                    \label{step:3}\\
          &\exn& \rhs\label{step:4}
        \end{eqnarray}
        Step~\rref{step:3} holds because $x$ does not occur in the
        two prefixes, and step~\rref{step:4} holds
        by \IHone.

        Since $M$ is a value, we have to prove $\new
        x(\encbv{\lambda z.M'}{y_1}|\fwd xy)\exn\encbv{\lambda
          z.{M'\subst xy}}{y_1}$. This follows by \IHone, along the
        lines of the above proof.
   \item If $M=\lambda x.M'$, then $x\notin\fnames{\encb Mp}$ and we
       can observe that:  first, the link $\fwd xy$ together with the
       restriction on $x$ can be erased up to $\sim$; second, $M\subst
       xy = M$. We can thus conclude.

       Again, $M$ is a value, so we need to prove the corresponding
       property, which is done as in the proof above.
      \end{itemize}

      \textbf{\Ltwo.} We know from the first case in the induction
      (step~\rref{step13}) 
      that
      \begin{eqnarray}
      \new p(\encb{\lambda z.M'}p | \fwd pq)
      &\exn & \bout q{y'}.\new y(\encbv{\lambda z.M'} y|\fwd
              {y'}y)\nonumber\\
        &=&
\bout q{y'}.\new
   y(!y(z,q').\encb{M'}{q'} | \fwd{y'}y)
            \nonumber\\
        &\sim&
         \bout q{y'}.!y'(z_1,q_1).\new y(!y(z,q').\encb{M'}{q'}
         \nonumber\\
        &&\qquad\qquad
        | \bout{y}{z,q'}.(\fwd z{z_1} | \fwd {q'}{q_1})) 
           \label{step14}\\
        &\exn&
               \bout q{y'}.!y'(z_1,q_1).(\new y,z,q')(!y(z,q').\encb{M'}{q'}
         \nonumber\\
        &&\qquad\qquad
           | \encb{M'}{q'} | \fwd z{z_1} | \fwd{q'}{q_1}) 
           \label{step15}\\
        &\sim&
               \bout q{y'}.!y'(z_1,q_1).(\new z,q')(\encb{M'}{q'}
               | \fwd z{z_1} | \fwd{q'}{q_1}) 
               \label{step16}\\
        &\exn&
               \bout q{y'}.!y'(z_1,q_1).\encb{M'\subst z{z_1}}{q_1}
               \label{step17}\\
        &=& \encb{\lambda z.M'}q
      \end{eqnarray}

      Step~\rref{step14} holds because the input at $y'$ is the only
      transition that can be performed after the bound output at $q$.
      Step~\rref{step15} holds by performing a deterministic
      communication on $y$.
      Step~\rref{step16} simply consists in garbage-collecting the
      input at $y$.
      For step~\rref{step17}, we use \IHone{} and \IHtwo{}| note that
      \IHtwo{} is necessary here to establish \IHone.
      
      \textbf{\Lthree.}
      We write
      \begin{eqnarray}
        \new y(!y(z,q).\encb{M'}q | \fwd xy)
        &\sim&
               !x(z',q').\new y(!y(z,q).\encb{M'}q
               \nonumber\\&&\qquad\qquad
                             | \bout
               y{z,q}.(\fwd z{z'}|\fwd q{q'}))
                             \nonumber\\
        &\exn&
               !x(z',q').(\new z,q)(\encb{M'}q | \fwd z{z'}|\fwd
        q{q'})
                             \nonumber\\
        &\exn&
               !x(z',q').\encb{M'\subst z{z'}}{q'}
                             \nonumber
      \end{eqnarray}

      The reasoning steps are like in the proof above: expand the
      behaviour of a link process, perform a deterministic
      communication, and rely on \IHone{} and \IHtwo{} to get rid of
      the forwarders.
      We note that \IHone{} and \IHtwo{} are used to prove \Lthree{}
      in this case.
\bigskip

\noindent\textbf{Third case: $M$ is an application.}

\noindent
\textbf{\Lthree.} We do not have to consider \Lthree{} in this case, since $M$ is not a value.

\medskip

    There are 5 sub-cases,
    according to the definition of the optimised encoding of
    Figure~\ref{f:opt_encod}.

We let \wires{} stand
    for the process $\bout {y_1}{w',r'}.(\fwd{w'}w | \fwd{r'}p)$, which is
    used in four of the clauses in the encoding for application.

    \medskip

    \noindent
    \textbf{\Ltwo.}
    To prove \Ltwo, we reason in the same way in four sub-cases, namely
    all except $M=z_1\val$: in these sub-cases the only occurrence of
    $p$ in $\encb Mp$ is in the sub-process \wires. We reason modulo
    strong bisimilarity ($\sim$) to bring the forwarder $\fwd pq$
    close to that occurrence, yielding a subterm of the form
    $\bout {y_1}{w',r'}.(\fwd{w'}w | \new p(\fwd{r'}p|\fwd pq))$. We
    use Lemma~\ref{l:fwd} to deduce that this process expands
    $\bout {y_1}{w',r'}.(\fwd{w'}w | \fwd{r'}q)$, which allows us to
    establish \Ltwo.

    Similarly, in the last case ($M=z_1\val$), we write
        $$\new p(\bout{z_1}{z,q'}.( \encbv{\val}z | \fwd {q'}p) |\fwd pq)
        ~\sim~
                 \bout{z_1}{z,q'}.( \encbv{\val}z | \new p(\fwd{q'}p
                 |\fwd pq)),$$
        and we conclude using Lemma~\ref{l:fwd}.

    \medskip

          \noindent\textbf{\Lone.}
    We analyse the 5 cases corresponding to the optimised encoding of application.
        \begin{itemize}
        \item $M=M'N'$, and none of $M'$ and $N'$ are values.

          Then
      \begin{eqnarray}
        \lhs
        & = &
              \new x(\new q(\encb{M'}q | q(y_1).\new r(\encb{N'}r |
              r(w).\wires)) | \fwd xy)
              \nonumber\\
        &\sim&
               \new q(\new x(\encb{M'}q | \fwd xy) | q(y_1).\new r(\new
               x(\encb{N'}r | \fwd xy) | r(w).\wires))
               \label{step:5}\\
        &\exn& \new q(\encb{M'\subst xy}q | q(y_1).\new r(\encb{N'\subst
               xy}r | r(w).\wires))\label{step:6} = \rhs
      \end{eqnarray}
      Step~\rref{step:5} holds by distributivity properties of private
      replications, and step~\rref{step:6} holds by applying 
      \IHone{} twice.


    \item $M=M'\val$.
      Then
      \begin{eqnarray}
          \lhs& = &
                    \new x(\new q(\encb{M'}q | q(y_1).\new w(\encbv{\val}w
                    | \wires)) | \fwd xy)
                    \nonumber\\
          &\sim& \new q(\new x(\encb{M'}q | \fwd xy) | q(y_1).\new
                 w(\new x(\encbv{\val}w | \fwd xy) | \wires)
                 \label{step:7}\\
        &\exn& \rhs\label{step:8}
      \end{eqnarray}
      Step~\rref{step:7} is proved using the distributivity properties
      of private replications; step~\rref{step:8} follows by \IHone.

        
    \item $M=\val M'$. This case is proved using the same kind of
      reasoning as the previous one.
      
    \item $M=(\lambda z.M')\val$.
      Then
      \begin{eqnarray}
        \lhs &=&
                 \new x(\new y_1,w(\encbv{\lambda z.M'}{y_1} |
                 \encbv{\val}w | \wires) | \fwd xy)
                 \nonumber\\
             &\sim&
                    \new y_1,w(\new x(\encbv{\lambda z.M'}{y_1}|\fwd
                 xy)
                    | \new x(\encbv{\val}w|\fwd xy) | \wires)
                    \nonumber\\
        &\exn &\rhs\nonumber
      \end{eqnarray}
      Here again, we distribute the forwarder and then apply \IHone{} twice (for the encoding
      $\encbv{\cdot}{\cdot}$).

    \item $M=z_1\val$.
      We have
      \begin{eqnarray}
        \lhs&=&
                \new x(\bout{z_1}{z,q}.( \encbv{\val}z | \fwd qp) |
                \fwd xy)
                \nonumber
      \end{eqnarray}
      We distinguish two cases:
      \begin{itemize}
      \item if $z_1\neq x$, $z_1\subst xy=z_1$, and
        \begin{eqnarray}
        \lhs& \sim&
                    \bout{z_1}{z,q}.\new x(\encbv{\val}z | \fwd xy | \fwd qp)
                    \nonumber\\
          &\exn& \bout{z_1}{z,q}.(\encbv{\val\subst xy}z | \fwd qp) =
             \rhs
                 \nonumber
        \end{eqnarray}
        Above, we apply \IHone{} for $\encbv{\val}z$.
      \item if $z_1=x$, then the only transition \lhs{} can make is a
        communication at $x$, so
        \begin{eqnarray}
          \lhs &\exn &
                       (\new x,q,z)(\encbv{\val}z | \fwd qp | \fwd xy |
                       \bout{y}{w_1,q_1}.(\fwd{q_1}q | \fwd{w_1}z))
                       \nonumber\\
          &\sim& \bout y{w_1,q_1}.\new x(\new z(\encbv{\val}z | \fwd {w_1}z) |
                 \new q(\fwd{q_1}q | \fwd qp) | \fwd xy)
                 \label{step:20}
        \end{eqnarray}
        Step~\rref{step:20} holds because the only transition the
        process can make is the bound output at $y$.

        We then use Lemma~\ref{l:fwd} to contract the forwarders,
        yielding $\fwd{q_1}p$; we use \IHthree{} to erase the
        forwarder $\fwd{w_1}z$, and we use \IHone{} to replace $\new
        x(\encbv\val{w_1} | \fwd xy)$ with $\encbv{\val\subst
          xy}{w_1}$. This finally yields \rhs.

        We remark here that we use \IHthree{} to show \Lone.
      \end{itemize}
    \end{itemize}
\end{proof}

\optsound*

\begin{proof}
We reason by induction over $M$.

\noindent
\textbf{Case $M=x$.} By definition, $\encb M p =\enc M p$.

\noindent
\textbf{Case $M=\abs x N$.} Assuming $\encb N q \exn\enc N q$, we have, by definition 
\begin{align*}
\encb M p &=\outb p y.!\inp y {x,q}.\encb N q\\
&\exn \outb p y .!\inp y{x,q}.\enc N q ~ = \enc M p
\end{align*}

\noindent
\textbf{Case $M=M_1 M_2$.} We have by definition
  \begin{eqnarray*}
    \enc{M_1 M_2}p &=& \new q(\enc{M_1}q | q(y).\new r(\enc{M_2}r |
                       r(w).\wires))
    \\
    && \qquad\mbox{with }\wires = \bout y{w',p'}.(\fwd{w'}w | \fwd{p'}p)
  \end{eqnarray*}
We distinguish 5 cases, according to the definition of the optimised
encoding in Figure~\ref{f:opt_encod}.
\begin{itemize}
\item $M_1=x, M_2=\val$.
\begin{eqnarray}
  \enc {x\val} p
    &\exn& \new q(\bout qy.\fwd yx|\inp q y.\new r(\encb \val r |\inp
           r w.\wires))
           \label{step:xV1}\\
    &\exn& \new {y,w}(\fwd y x|\encbv \val w | \bout
           y{w',p'}.(\fwd{w'}w |\fwd{p'}p))
           \nonumber\\
&\exn& \new{w,w'}(\outb x {z,q}.(\fwd z{w'}|\fwd q{p'})|\encbv \val
       w|\fwd{w'}w|\fwd{p'}p)
  \nonumber\\
&\sim & \outb x {z,q}.\new{w,w'}(\encbv \val w|\fwd
        z{w'}|\fwd{w'}w|\fwd q{p'}|\fwd{p'}p)
  \label{step22}\\
    &\exn& \outb x {z,q}.\new w(\encbv \val w |\fwd z w|\fwd q p)
  \label{step23}\\
    &\exn& \outb x {z,q}.(\encbv \val z|\fwd q p) ~=~ \encb{x\val}p
           \label{step24}
\end{eqnarray}
Step~\rref{step:xV1} follows by definition (for the encoding of $x$), and by induction (for
the encoding of $\val$).
The two following $\exn$ steps are derived by performing deterministic
$\tau$ transitions.
We then remark that the only action that can be performed is a bound
output at $x$ (step~\rref{step22}), and contract fowarders using
Lemma~\ref{l:fwd} (step~\rref{step23}).
Finally, we use Lemma~\ref{l:fwd:optim} in step~\rref{step24}.
\item Case $M=(\abs x N) \val$.
\begin{eqnarray}
  \enc {(\abs x N )\val} p
  &=& \new q (\outb q y.!\inp y{x,p'}.\enc N{p'}
      |\inp q y.\new r(\enc
      \val r|
    \inp r w .\wires)))\nonumber\\ 
  &\exn& \new q (\outb q y.!\inp y{x,p'}.\enc N{p'}\nonumber\\
  &&\qquad |\inp q y.
     \new r(\outb r w.\encbv \val w|\inp r w .\wires))
  \label{step:beta1}\\
  &\exn& \new {(y,w)}(!\inp y {x,p'}.\enc N {p'}|\encbv \val w| \wires)
  \label{step:beta2}\\
  &\exn& \new {(y,w)}(!\inp y {x,p'}.\encb N {p'}|\encbv \val w| \wires)
  \label{step:beta3}\\
&=& \encb M p\nonumber
\end{eqnarray}
Steps~\rref{step:beta1} and~\rref{step:beta3} follow by
induction. Step~\rref{step:beta2} follows by performing two
deterministic $\tau$-transitions and garbage-collecting the
restrictions on $q$ and $r$.
\item Case $M_1=\val, M_2=N$.
\begin{eqnarray*}
  \enc{\val N} p
&\exn&\new q(\outb q y. \encbv \val y |\inp q y.\new r (\encb N r|\inp
       r w.\wires ))\\
    &\exn &\new y(\encbv \val y |\new r (\encb N r|\inp r w.\wires))
~ =\encb {\val N}p
\end{eqnarray*}
Again, we use the inductive hypothesis, and perform a
deterministic $\tau$-transition on $q$.
\item Case $M_1=N, M_2=\val$.
\begin{eqnarray*}
  \enc{N \val} p
&\exn&\new q(\encb N q |\inp q y.\new r (\outb r w.\encbv\val r|\inp r w.\wires))\\
&\exn& \new q(\encb N q |\inp q y.\new w (\encbv\val w|\wires))
~=\encb { N\val}p
\end{eqnarray*}
Again, we first use the inductive hypothesis, then contract a
deterministic communication on $r$.
\item Finally, if neither $M_1$ nor $M_2$ is a value, the two
  encodings coincide, and the property is immediate.
\end{itemize}
\end{proof}

\subsection{Operational Correspondence and Soundness}\label{a:soundness}

The following lemma is the central property we need to derive the validity
of $\betav$-reduction. 

\begin{lemma}\label{l:subst:value:optim}
  $\new x(\encb Mp |\encbv\val x)\exn \encb {M\subst x \val} p$.
\end{lemma}
\begin{proof}
  We establish the following property

  \begin{quotation}
    $\new x(\encb Mp |\encbv\val x)\exn \encb {M\subst x \val} p$

    and, if $M$ is a value $\val'$,
    $\new x(\encbv{\val'}y | \encbv{\val}x \exn \encbv{\val'\subst
      x\val}y$
  \end{quotation}
   by induction over the size of $M$. We write \rhs{} for the right hand
  side of the first relation above, that is, \rhs{} stands for
  $\encb{M\subst x\val}p$. We use similarly \lhs{} for
  $\new x(\encb Mp |\encbv\val x)$.

  We distinguish several cases, following the definition of the
  optimised encoding in Figure~\ref{f:opt_encod}.
  
  \begin{itemize}
  \item \textbf{$M$ is a variable.} We distinguish two sub-cases.
    \begin{itemize}
    \item $M=z, z\neq x$. Then $M\subst x\val = z$, and we can write
      \begin{eqnarray}
        \lhs & = &\new x(\bout py.\fwd yz | \encbv{\val}x)\nonumber\\
             &\sim& \bout py.\fwd yz\label{step:sim0}\\
        &=&\rhs\nonumber
      \end{eqnarray}
      Relation~\rref{step:sim0} above holds because $x$ is
      fresh for $\bout py.\fwd yz$ and because $\encbv{\val} x$ starts
      with an input on $x$.

      Since $M$ is a value, we must also prove the second relation
      mentioned above. We have indeed
      $\new x(\encbv zy|\encbv{\val}x) \sim \fwd yz=\encbv{\val}x$,
      and we can conclude since $\sim\protect{\subseteq}\exn$.
      
    \item $M=x$. Then $M\subst x\val = \val$, and we can write
      \begin{eqnarray}
       \lhs & = & \new x(\bout py.\fwd yx | \encbv{\val}x) \nonumber\\
             & \sim & \bout py.\new x(\encbv{\val}x | \fwd yx)\label{step:sim1}\\
             & \exn& \bout py.\encbv{\val}y \label{step:exn1}\\
        &=&\rhs \nonumber 
      \end{eqnarray}
      Relation~\rref{step:sim1} holds because $\encbv{\val}x$ starts
      with an input at $x$. Relation~\rref{step:exn1} follows by
      Lemma~\ref{l:fwd:optim}.

      Again, in this case $M$ is a value, so we must also prove the
      second relation. We can indeed check that $\new
      x(\encbv{x}y |\encbv{\val}x) \exn \encbv{\val}y$ by
      Lemma~\ref{l:fwd:optim}. 
    \end{itemize}

  \item \textbf{$M$ is an abstraction.} We distinguish two sub-cases.
    \begin{itemize}
    \item $M=\lambda x.M'$. Then $M\subst x\val = M$, and we can write
      \begin{eqnarray}
        \lhs & = & \new x(\bout py.!y(x,q).\encb{M'}q |
   \encbv{\val}x)\nonumber\\
             &\sim& \bout py.!y(x,q).\encb{M'}q\label{step:sim00}\\
        & = & \rhs\nonumber 
      \end{eqnarray}
      Relation~\rref{step:sim00} above holds because $x$ does not occur free in
      $\bout py.!y(x,q)\encb{M'}q$, and $\encbv{\val}x$ starts with an
      input at $x$.

      $M$ is a value, so we must prove the second relation as
      well. We have
      \begin{eqnarray}
        &&\new x(\encbv{\lambda x.M'}y | \encbv{\val}x)
        \nonumber\\&=&
                       \new x(!y(x,q).\encb{M'}q | \encbv{\val}x)
        \nonumber\\&\sim&
                          !y(x,q).\encb{M'}q | \new x\encbv{\val}x
                          \label{l:step:sim001}\\&\sim&
                          \encbv{M}y
                          \label{l:step:sim002}
      \end{eqnarray}
      Relation~\rref{l:step:sim001} holds because $x$ does not occur
      free in $!y(x,q).\encb{M'}q$, and relation~\rref{l:step:sim002}
      holds because the only possible transition of $\encbv{\val}x$ is
       an input at $x$.
    \item $M=\lambda z.M'$ with $z\neq x$. Then $M\subst x\val =
      \lambda z.(M'\subst x\val)$, and by definition, 
      $      \rhs = \bout py.!y(z,q).\encb{M'\subst x\val}q      $.
      
      We can write
      \begin{eqnarray*}
        \lhs & = & \new x(\bout py.!y(z,q).\encb{M'}q | \encbv{\val}x)\\
             &\sim & \bout py.!y(z,q).\new x(\encb{M'}q | \encbv{\val}x)\\
        & \exn& \bout py.!y(z,q).\encbv{M'\subst x\val}q = \rhs
      \end{eqnarray*}

      $M$ is a value, so we must also prove the second relation:
      \begin{eqnarray}
        &&\new x(\encbv{\lambda z.M'}y | \encbv{\val}x)\nonumber\\
        &=&
            \new x(!y(z,q).\encb{M'}q | \encbv{\val}x)\nonumber\\
        &\sim&
               !y(z,q).\new x(\encb{M'}q | \encbv{\val}x)\\
        &\exn& !y(z,q).\encb{M'\subst{\val}x}q\\
        &=& \encbv{M\subst{\val}x}y\nonumber               
      \end{eqnarray}

    \end{itemize}
  \item \textbf{$M$ is an application.} We distinguish 5 sub-cases,
    according to the definition of the optimised encoding of
    Figure~\ref{f:opt_encod}. In the following, we let \wires{} stand
    for the process $\bout y{w',r'}.(\fwd{w'}w | \fwd{r'}p)$, which is
    used in the different clauses in the encoding.

    We also make use of some standard \emph{properties of replicated
    resources}~\cite{SW01a}. 
    \begin{itemize}
    \item $M=M'N'$, and none of $M'$ and $N'$ are values.
      Then we have:
      \begin{eqnarray}
        \lhs &=&
                    \new x\new q(\encb Mq | \encbv{\val}x
                    | q(y).\new r(\encb Nr | r(w).\wires) )
        \nonumber\\
        &\sim& 
               \new q\new x(\encb Mq | \encbv{\val}x
               \nonumber\\&&\qquad
                    | q(y).\new r(\new x(\encb Nr | \encbv{\val}x) |
 r(w).\wires) )\label{step:bis2}
        \\
             &\exn&
                    \new q(\encb{M\subst x\val}q
                    | q(y).\new r(\encb{N\subst x\val}r | r(w).\wires)
                   )\label{step:exn2}
                    \\ &=&\rhs\nonumber
      \end{eqnarray}
      Relation~\rref{step:bis2} holds by the distributivity properties
      of private replications~\cite{SW01a} (note in particular that $x$ is not used
      in input in $\encb Mq$ and in $\encb
      Nr$). Relation~\rref{step:exn2} holds by using the induction
      hypothesis twice.
      
    \item $M=M'\val'$. Then
      \begin{eqnarray}
        \lhs &=& \new x(\new q(\encb{M'}q | q(y).\new w(\encbv{\val'}w |
         \wires) | \encbv{\val}x)
                 \nonumber\\
             &\sim&        
                    \new q(\new x(\encb{M'}q | \encbv{\val}x)
                    \nonumber \\
             &&\label{step:sim3}
                \qquad| q(y).\new w(\new x(\encbv{\val'}w | \encbv{\val}x) |
                    \wires) ))
                    \\
             &\exn&        \label{step:exn3}
                    \new q(\encb{M'\subst{\val}x}q
                    | q(y).\new w(\encbv{\val'\subst{\val}x}w |
                    \wires) )
        \\&=&\rhs\nonumber
      \end{eqnarray}
      
    \item $M=\val'M'$. Then
      \begin{eqnarray}
        \lhs & = & \new x(\new y(\encbv{\val'}y | \new r(\encb{M'}r |
               r(w).\wires | \encbv{\val}x) \nonumber\\
             &\sim&
                    \new y(\new x(\encbv{\val'}y | \encbv{\val}x)
           \nonumber
        \\
        && \qquad| \new r(\new x(\encb{M'}r |\encbv{\val}x) | \wires) )))
           \label{l:step:bis4}
        \\
             &\exn& \label{l:step:exn4}
                    \new r(\encb{M'\subst{\val} x}r
                    | \new r(\encbv{\val'\subst{\val}x}r | \wires))
      \end{eqnarray}

    \item $M=x'\val'$.
      \begin{itemize}
      \item If $x'=x$, then $M\subst{\val}x = \val~\val'$.
        We have
        \begin{eqnarray}
          \lhs &=&\nonumber
                   \new x(\bout x{z,q}.(\encbv{\val'}z | \fwd qp) |
     \encbv{\val}x)
        \end{eqnarray}
        We consider two cases.

        Suppose $\val = z$, then $\encbv{\val}x = \fwd xz$.

        Suppose $\val = \lambda z.M'$, then $\encbv{\val}x = !x(z,q).\encb{M'}q$.
        
      \item If $x'\neq x$, then $M\subst{\val}x = x'~\val'$.
        We have
        \begin{eqnarray}
          \lhs &=&\nonumber
                   \new x(\bout {x'}{z,q}.(\encbv{\val'}z | \fwd qp) |
     \encbv{\val}x)
          \\
          &\sim& \label{step:sim5}
                   \bout {x'}{z,q}.(\new x(\encbv{\val'}z |
                 \encbv{\val}x) | \fwd qp)
          \\
               &\exn&\label{step:exn5}
                   \bout {x'}{z,q}.(\encbv{\val'\subst{\val}x}z 
                      | \fwd qp) = \rhs
        \end{eqnarray}

      \end{itemize}

    \item $M=(\lambda x'.M')\val'$. Then
      \begin{eqnarray}
        \lhs &=& \new x(\new y,w(\encbv{\lambda x'.M'}y |
                 \encbv{\val'}w | \wires) | \encbv{\val}x)
                 \nonumber\\
        &\sim & \new y,w(\new x(\encbv{\lambda x'.M'}y |
                \encbv{\val}x)\nonumber\\
        &&\qquad | \new x(\encbv{\val'}w |
                \encbv{\val}x) | \wires) \label{step:bis6}\\
        &\exn & \new y,w(\new x(\encbv{(\lambda x'.M'}y)\subst{\val}x |
                \encbv{\val}x)\nonumber\\
        && \qquad | \new x(\encbv{\val'\subst{\val}x}w |
                \encbv{\val}x) | \wires) \label{step:exn6}\\
        &=&\rhs\nonumber
      \end{eqnarray}

    \end{itemize}

  \end{itemize}
\end{proof}

\betaval*
\begin{proof}
  We show a stronger property, namely that $\encb
  Mp\arr\tau\exn\encb Np$.
  This is a consequence of Lemma~\ref{l:subst:value:optim}, exploiting the
  congruence properties of $\exn$, and the fact that for any $M$,
  $\evctxt$ and $p$, the only transition
  $\encb{\evctxt[M]}p$ can do arises from a transition of the encoding
  of $M$ (intuitively, the encoding of the hole is in active position).
\end{proof}

\optstuck*
\begin{proof}
  We reason by induction over the shape of the evaluation context $\evctxt$.
  \begin{itemize}
  \item \textbf{Base case:} $\evctxt=\hdot$. We
    observe $\encb yp = \outb pz.\fwd zy$, hence $q(y).\encb yp=\fwd qp$.

    We then write
    \begin{eqnarray*}
      \encb {x\val} p & = & \outb x{z,q}.(\encbv\val z |\fwd qp)
                            \mbox{ by definition}
      \\
      &=& \outb x {z,q}.(\encbv \val z|\inp q
y.\encb {y}p)
    \end{eqnarray*}
  \item \textbf{Case }$\evctxt=\val'\evctxt'$.
    We write
    \begin{eqnarray}\nonumber
      \encb {\evctxt[x\val]} p
      &=& 
       \new s\big(
      \encbv{\val'}s | \new r(\encb{\evctxt'[x\val]}r | P_0) \big)
      \\ &&\nonumber\qquad
            \mbox{with }P_0 = {r(w).\bout{s}{w',r'}.(\fwd{w'}w |
            \fwd{r'}p)}
    \end{eqnarray}

     We have by induction
    \begin{mathpar}
      \encb{\evctxt'[x\val]}r
      \sim
      \bout{x}{z_1,q_1}.(\encbv{\val}{z_1} 
      | q_1(y_1).\encb{\evctxt'[y_1]}r),
    \end{mathpar}
    which gives
    \begin{eqnarray}
      \encb {\evctxt[x\val]} p
      &\sim& 
       \new s\big(
               \encbv{\val'}s | \nonumber\\
      && \qquad\new r(
               \bout{x}{z_1,q_1}.(\encbv{\val}{z_1} 
                  | q_1(y_1).\encb{\evctxt'[y_1]}r)
         | P_0) \big)
         \nonumber
      \\
      &\sim&
\bout{x}{z_1,q_1}.
       \new s\big(
             \encbv{\val'}s \nonumber\\
      &&\qquad | \new r(
             \encbv{\val}{z_1} 
                  | q_1(y_1).\encb{\evctxt'[y_1]}r
             | P_0) \big)
      \label{eq:VC:last}
               \\&\sim&
\bout{x}{z_1,q_1}.
             \big(
             \encbv{\val}{z_1}
             \nonumber\\&&\qquad
             |
             q_1(y_1).(
             \new s(
             \encbv{\val'}s | \new r(
             \encb{\evctxt'[y_1]}r
                           | P_0))) \big)
                           \label{eq:VC:last1}
    \end{eqnarray}
    
For~\rref{eq:VC:last}, we observe that $\encbv{\val'}s$ starts with an input at $s$, and
$P_0$ starts with an input at $r$. Therefore, 
the bound output at $x$ is the only 
possible transition for the process above, which
allows us to bring the prefix on top.

For~\rref{eq:VC:last1}, we recall that
    $P_0 = {r(w).\bout{s}{w',r'}.(\fwd{w'}w | \fwd{r'}p)}$.
    We observe that $P_0$ can start interacting only after the prefix
    $q_1(y_1)$ is triggered, because the only possible output at $r$
    is within $\encb{\evctxt'[y_1]}r$.
    In turn, because the output at $s$ in
    $P_0$ is guarded by the input at $r$, the subterm $\encbv{\val'}s$
    can become active only after the interaction at $r$, and hence it
    is sound, modulo strong bisimilarity, to place $\encbv{\val'}s$
    under the prefix $q_1(y_1)$.
 
    We can then conclude, by observing that
    $ \new s( \encbv{\val'}s | \new r( \encb{\evctxt'[y_1]}r | P_0))$
    is equal to $\encb{\val'\evctxt'[y_1]}p$ by definition.

  \item \textbf{Case} $\evctxt = \evctxt' M$. We reason as follows:
    \begin{eqnarray}
      \encb{\evctxt[x\val]}p
      &=&
          \new s(\encb{\evctxt'[x\val]}s
          | s(z).\new r(\encb Mr|P_0) )
          \nonumber
      \\&&
           \qquad\mbox{ with }P_0=r(w).\bout{z}{w',r'}.(\fwd{w'}w |
 \fwd{r'}p)
           \nonumber
      \\&\sim& \label{eq:CVT:ind}
               \new s\big(\,
               \bout{x}{z_1,q_1}.(\encbv{\val}{z_1}
               | q_1(y_1).\encb{\evctxt'[y_1]}s)
               \nonumber\\
      &&\qquad | s(z).\new r(\encb Mr|P_0)\,\big)
      \\&\sim& \label{eq:CVT:sim}
               \bout{x}{z_1,q_1}.
               \big(\,
\encbv{\val}{z_1} | q_1(y_1).\new s(\encb{\evctxt'[y_1]}s)
               \nonumber\\
      &&\qquad
         | s(z).\new r(\encb Mr|P_0))\,\big)
    \end{eqnarray}
Step~(\ref{eq:CVT:ind}) follows by induction, and
step~(\ref{eq:CVT:sim}) is deduced as in the previous case.
   \end{itemize}
\end{proof}

\subsection{Completeness}

\solaux*

\begin{proof}
We use Lemma~\ref{l:opt_stuck} and~\ref{l:opt_sound}; we get:
$$\enc {\evctxt[x\val]} p\bsim \outb x {z,q}.(\encbv \val z|\inp q y.\enc {\evctxt[y]}p)$$
and 
$$\enc {(\abs w\evctxt[w])(x\val)} p\bsim \outb x {z,q}.(\encbv \val z|\inp q y.\enc {(\abs w\evctxt[w])y}p)$$
We conclude by validity of $\beta$-reduction (Lemma~\ref{l:beta}) applied to $\enc{(\abs w\evctxt[w])y}p$.
%
%
\end{proof}










%

\section{Systems of equations for \alpi{} (Section~\ref{s:localpi})}\label{a:alpi}

The systems of
equations for \alpi{} are presented on
Figures~\ref{f:eqalpi} and~\ref{f:optalpi}.

To introduce the second system of equations, we define the extension
of the encoding to equation variables as follows:
\[ 
 \encm {X_{M,N}}{} \defi \bind {p}  \app{X_{M,N}}{\tily,p} 
\mbox{ ~~where $\tily = \fv{M,N}$}
\]

\begin{figure*}[t]
  \centering
\begin{align*}
&M\diverges\text{ and } N\diverges:
&X_{M,N} &= \bind {\til y} \encm \Omega \\
%
&M\converges x \text{ and }N\converges x:
&X_{M,N}&=\bind {\til y} \encm x \\
%
&M\converges \abs x M'\text{ and } N\converges \abs x N':
&X_{M,N}&=\bind {\til y}
\encm{\abs x X_{M',N'}
}\\
%
&M\converges \evctxt[x\val]\text{ and }N\converges\evctxt'[x\valp]:
&X_{M,N} &= \bind {\til y} 
\encm {(\abs z X_{\evctxt[z],\evctxt'[z]}
)
~(x~X_{\val,\valp}
)} \\
%
&M\converges x\text{, }N\converges \abs z N'\text{, }N'\converges\evctxt[x\val]:
&X_{M,N} &=\bind {\til y}
\encm{\abs z \left((\abs w X_{w,\evctxt[w]}) ~ (x~X_{z,\val}
)\right)}\\
%
&M\converges \abs z M'\text{, }M'\converges \evctxt[x\val]\text{, }N\converges x:
&X_{M,N} &=\bind {\til y}
\encm{\abs z \left((\abs w X_{\evctxt[w],w}) ~ (x~X_{\val,z}
)\right)}
\end{align*}
\caption{System \eqalpi of equations
 (the last two equations are only needed for $\enfe$)
}
\label{f:eqalpi}
\end{figure*}

\begin{figure*}[t]
  \centering
\begin{align*}
&M\diverges\text{ and } N\diverges:
&X_{M,N} &= \bind {\til y,p} \zero \\
%
&M\converges \evctxt[xv]\text{ and }N\converges\evctxt'[xv']:
&X_{M,N} &= \bind {\til y,p} (\new{z,q})(\out x {z,q}|
\XV_{\val,\valp}\param{z,{\tilprime y}}
                   \\ &&&\qquad
|\inp q
  w.X_{\evctxt[w],\evctxt'[w]}\param{{\tilpprime y},p}) 
\\
%
&M\converges \val\text{ and }N\converges \valp:
&X_{M,N} &= \bind{\til y,p} (\new y)(\out p y|\XV_{v,v'}\param{z,\tilprime y})\\
%
&\val= x \text{ and }\valp= x:
&\XV_{x,x}&=\bind{z,x}\alpilink z x \\
%
&\val=  \abs x M\text{ and } \valp= \abs x N:
&\XV_{\abs x M,\abs xN}&=\bind {z,\til y}
!\inp z {x,q}.X_{M,N}\param{\tilprime y ,q}\\
%
&\val=x\text{, }\valp= \abs z N\text{, }N\converges\evctxt[x\val]:
&       
\XV_{x,\abs z N}
&=\bind {y_0,\til y}
                   !\inp {y_0}{z,q}. (\new{z',q'})
                   \\ &&&\qquad
                   (\out x {z',q'}| \XV_{z,\val}\param{z',\tilprime y }
                   \\ &&&\qquad\quad
        |\inp {q'}{w}.X_{w,\evctxt[w]}\param{\tilpprime y,q})\\
%
&\val=\abs z M\text{, }M\converges\evctxt[x\val]\text{, }\valp= x:
&       
\XV_{\abs z M,x}
&=\bind {y_0,\til y}
                   !\inp {y_0}{z,q}. (\new{z',q'})
                   \\ &&&\qquad
(\out x {z',q'}| \XV_{\val,z}\param{z',\tilprime y }
                   \\ &&&\qquad\quad
        |\inp {q'}{w}.X_{\evctxt[w],w}\param{\tilpprime y,q})
\end{align*}

\caption{System \eqalpip of equations  
(the last two equations are only needed for $\enfe$)}
\label{f:optalpi}
\end{figure*}


\section{Unique solution techniques for contextual relations (Section~\ref{s:contextual})}
\label{a:usoltrace}
\label{a:trace}


The proof of the following lemma is very similar to the proof of
Theorem~\ref{thm:usol}.  For more details, we refer the reader to
\cite{usol}, particularly the proof of unique solution for weak
bisimilarity in the setting of CCS.

\newcommand{\sh}{\param{\til a}}

\usoltrace*
\begin{proof}
For simplicity, we only give 
the proof for a single equation $E$, rather than a system of equations. 
Generalisation to systems of equations does not add any particular difficulty. 

Assume $E$ is an equation, $F$ an abstraction, and $F\trincl E[P]$. We 
fix a set of fresh names $\til a$, and write $P$ for $F\sh$. 
If $\til\alpha=\alpha_1\dots\alpha_n$ is a {finite} 
trace of $P$, 
we build a growing sequence of transitions 
of $E^n\sh$ such that 
$E^n[F]\sh\Arr{\alpha_1\dots\alpha_{i_k}} E_n[F]
\Arr{\alpha_{i_{k}+1}, \dots,\alpha_{n}}P_n$.

We start by making two observations:
\begin{enumerate}
\item If the transitions in $\til\alpha$
are all transitions of the context $E^n\sh$, we stop and we have
$E^n\sh\Arr{\til{\alpha}}$, and thus $K_E\sh\Arr{\til{\alpha}}$ . So $\til{\alpha}$ is a trace of $K_E\sh$.\label{item:trace1}
\item Otherwise there is an infinite sequence of transitions 
from $K_E\sh$ with visible actions $\lsnn \alpha {i_k}$ for some $k$; 
therefore $K_E\sh$ has a divergence.\label{item:trace2}
\end{enumerate}

We now explain the construction of the sequence.
Assume for that that we have both
$(i):~E^n\sh \Arr{\alpha_1,\dots,\alpha_{i_k}}E_n$ and
$(ii):~E_n[F]\Arr{\alpha_{i_{k}+1},\dots,\alpha_n}$.

By $(i)$ it follows that $E^{n+1}[F]\sh\Arr{\alpha_1,\dots,\alpha_{i_{k}}}E_n[E[F]]$.

By $(ii)$ and congruence of $\trincl$, it follows that
$\alpha_{i_{k}+1},\dots,\alpha_n$ is a trace of $E_n\sh[F]\trincl
E_n[E[F]]$. 

We take for the new sequence of transitions the concatenation of 
 the previous one, and the part of $E_n[E[F]]\Arr{\alpha_{i_{k}+1},\dots,\alpha_n}$ 
 that is a transition of the context $E_n[E]$. Since $E$ is weakly guarded, this 
 is not an empty sequence.

By observation~\ref{item:trace2} above this construction has to stop, otherwise there would 
be a divergence. We conclude by observation~\ref{item:trace1}.



\end{proof}


\end{document}
\newpage
\tableofcontents
\newpage

\ifcomments
\clearpage
\section{New TODO}
\input{newtodo}

\input{todo}

\clearpage
\section{Removed from conclusion}
\input{conclusion-removed}

\clearpage
\section{Old appendices, to be removed}


\section{Properties of Milner's encodings}
\label{a:encodings}
\input{long}


\clearpage
\section{More details}
\input{details}

\clearpage
\section{Things removed (for now)}
\input{removed}
\section{Properties of Milner's encodings}
\DS{as far as i can see, ALL this appendix could be removed\\
DH: maybe the only thing that can be saved is the technical
explanation of why law (1) fails: put it in the main text?}
\input{long}
\fi

\end{document}